\documentclass{article}
\usepackage{amsmath,amssymb,amsfonts}

\newcommand{\beq}{\begin{equation}}
\newcommand{\eeq}{\end{equation}}
\newcommand{\ben}{\begin{eqnarray}}
\newcommand{\een}{\end{eqnarray}}
\newcommand{\be}{\begin{eqnarray*}}
\newcommand{\ee}{\end{eqnarray*}}

\newcounter{eqalph}
\newcounter{equationa}

\let\ssection=\section
\renewcommand{\section}{\setcounter{equation}{0}\ssection}

\newcounter{example}[section]
\newcounter{remark}[section]
\newcounter{theorem}[section]
\newcounter{proposition}[section]
\newcounter{lemma}[section]
\newcounter{corollary}[section]
\newcounter{definition}[section]

\def\theremark{\arabic{section}.\arabic{remark}}
\def\thetheorem{\arabic{section}.\arabic{theorem}}

\def\thedefinition{\arabic{section}.\arabic{definition}}

\newenvironment{proof}{\noindent {\textit{Proof:}}
}{$\Box$\medskip}
\newenvironment{example}{\refstepcounter{remark}\medskip\noindent{\bf
Example \theremark:} }{$\Box$ \medskip}
\newenvironment{remark}{\refstepcounter{remark}\medskip\noindent{\bf
Remark \theremark:} }{$\Box$\medskip}
\newenvironment{theorem}{\refstepcounter{theorem}
\medskip\noindent{\sc Theorem \thetheorem}:}{$\Box$\medskip}

\newenvironment{lemma}{\refstepcounter{theorem}\medskip\noindent{\sc
Lemma \thetheorem}:}{ $\Box$\medskip }
\newenvironment{corollary}{\refstepcounter{theorem}\medskip\noindent{\sc
Corollary \thetheorem}:}{$\Box$\medskip}
\newenvironment{definition}{\refstepcounter{definition}\medskip\noindent{\sc
Definition \thedefinition}:}{$\Box$\medskip}

\def\op#1{\mathop{{\it\fam0} #1}\limits}

\newcommand{\wh}{\widehat}
\newcommand{\lrr}{\leftrightarrow}
\newcommand{\wt}{\widetilde}
\newcommand{\ol}{\overline}
\newcommand{\ot}{\otimes}

\newcommand{\ar}{\op\longrightarrow}
\newcommand{\id}{{\mathrm{Id}\,}}

\newcommand{\cC}{{\mathcal{C}\ell}}
\newcommand{\lC}{{\mathcal C}}
\newcommand{\cA}{{\mathcal A}}
\newcommand{\cV}{{\mathcal V}}

\newcommand{\cG}{{\mathcal G}}
\newcommand{\cK}{{\mathcal K}}
\newcommand{\cZ}{{\mathcal Z}}

\newcommand{\f}{\phi}
\newcommand{\Si}{\Sigma}
\newcommand{\si}{\sigma}
\newcommand{\vr}{\varrho}
\newcommand{\al}{\alpha}
\newcommand{\thh}{\theta}
\newcommand{\bt}{\beta}
\newcommand{\g}{\gamma}

\newcommand{\m}{\mu}
\newcommand{\la}{\lambda}
\newcommand{\bb}{{\mathbf 1}}
\newcommand{\dr}{\partial}

\newcommand{\rL}{{\mathrm L}}

\newcommand{\mar}[1]{}

\begin{document}

\hbox{}

\begin{center}

{\Large\bf Composite bundles in Clifford algebras. \\
Gravitation theory. Part I}
\bigskip
\bigskip

{\sc G. SARDANASHVILY, A. YARYGIN}
\bigskip

Department of Theoretical Physics, Moscow State University, Russia

\bigskip
\bigskip

\textbf{Abstract}
\end{center}

\noindent Based on a fact that complex Clifford algebras of even
dimension are isomorphic to the matrix ones, we consider bundles
in Clifford algebras whose structure group is a general linear
group acting on a Clifford algebra by left multiplications, but
not a group of its automorphisms. It is essential that such a
Clifford algebra bundle contains spinor subbundles, and that it
can be associated to a tangent bundle over a smooth manifold. This
is just the case of gravitation theory. However, different these
bundles need not be isomorphic. To characterize all of them, we
follow the technique of composite bundles. In gravitation theory,
this technique enables us to describe different types of spinor
fields in the presence of general linear connections and under
general covariant transformations.

\bigskip

\tableofcontents

\section{Introduction}

In this work, we aim to describe spinor fields in gravitation
theory in terms of bundles in Clifford algebras. A problem is that
gauge symmetries of gravitation theory are general covariant
transformations whereas spinor fields carry out representations of
Spin groups which are two-fold covers of the pseudo-orthogonal
ones.

In classical gauge theory, a case of matter fields which admit
only a subgroup of a gauge group is characterized as a spontaneous
symmetry breaking \cite{book09,sard08a,tmp}. Spontaneous symmetry
breaking is a quantum phenomenon, but it is characterized by a
classical background Higgs field \cite{sard08a,sard15}. In
classical gauge theory on a principal bundle $P\to X$ with a
structure Lie group $G$, spontaneous symmetry breaking is defines
as a reduction of this group to its closed (consequently, Lie)
subgroup $H$ (Definition \ref{ss41}). By virtue of the well-known
Theorem \ref{redsub}, there is one-to-one correspondence between
the $H$-principal subbundles $P^h$ of $P$ and the global sections
of the quotient bundle $\Si=P/H\to X$ (\ref{ss41}) with a typical
fibre $G/H$. These sections are treated as classical Higgs fields
\cite{book09,higgs,tmp}. Matter fields possessing only exact
symmetry group $H$  are described in a pair with Higgs fields as
sections of composite bundles $Y\to \Si \to X$ \cite{higgs13,tmp}.

This is just the case of Dirac spinor fields in gravitation theory
\cite{book09,sard98a,sard11}.

Theory of classical fields on a smooth manifold $X$ admits a
comprehensive mathematical formulation in the geometric terms of
smooth fibre bundles over $X$ \cite{book09,book13}. For instance,
Yang -- Mills gauge theory is theory of principal connections on a
principal bundle $P\to X$ with some structure Lie group $G$. Gauge
gravitation theory is formulated in the terms of fibre bundles
which belongs to the category of natural bundles
\cite{book09,sard11}.

Studying gauge gravitation theory, one requires that it
incorporates Einstein's General Relativity and, therefore, it
should be based on Relativity and Equivalence Principles
reformulated in the fibre bundle terms \cite{iva,sardz}. As a
consequence, gravitation theory has been formulated as gauge
theory of general covariant transformations with a Lorentz reduced
structure where a pseudo-Riemannian metric gravitational field is
treated as the corresponding classical Higgs field
\cite{iva,sardz,sard02,sard06}.

Relativity Principle states that gauge symmetries of classical
gravitation theory are general covariant transformations. Fibre
bundles possessing general covariant transformations constitute
the category of so called natural bundles \cite{book09,kol,terng}.

\begin{remark} \label{ss45} \mar{ss45}
Let $\pi:Y\to X$ be a smooth fibre bundle. Any automorphism
$(\Phi,f)$ of $Y$, by definition, is projected as $\pi \circ \Phi=
f\circ \pi$ onto a diffeomorphism $f$ of its base $X$. The
converse need not be true. A fibre bundle $Y\to X$ is called the
natural bundle if there exists a monomorphism
\be
{\rm Diff} X\ni f\to\wt f\in {\rm Aut} Y
\ee
of a group of diffeomorphisms of $X$ to a group of bundle
automorphisms of $Y\to X$. Automorphisms $\wt f$ are called
general covariant transformations of $Y$. The tangent bundle $TX$
of $X$ exemplifies a natural bundle. Any diffeomorphism $f$ of $X$
gives rise to the tangent automorphisms $\wt f=Tf$ of $TX$ which
is a general covariant transformation of $TX$. The associated
principal bundle is a fibre bundle $LX$ of linear frames in
tangent spaces to $X$ (Section 5.2). It also is a natural bundle.
Moreover, all fibre bundles associated to $LX$ are natural
bundles.
\end{remark}

Following Relativity Principle, one thus should develop
gravitation theory as a gauge theory on a principal frame bundle
$LX$ over an oriented four-dimensional smooth manifold $X$, called
the world manifold $X$ \cite{book09,sard11}.

Equivalence Principle reformulated in geometric terms requires
that the structure group
\mar{gl4}\beq
GL_4=GL^+(4, \mathbb R) \label{gl4}
\eeq
of a frame bundle $LX$ and associated bundles is reducible to a
Lorentz group $SO(1,3)$. It means that these fibre bundles admit
atlases with $SO(1,3)$-valued transition functions and,
equivalently, that there exist principal subbundles of $LX$ with a
Lorentz structure group (Section 5.2). This is just a case of
spontaneous symmetry breaking. Accordingly,  there is one-to-one
correspondence between the Lorentz principal subbundles of a frame
bundle $LX$ (called the Lorentz reduced structures) and the global
sections of the quotient bundle $LX/SO(1,3)\to X$ (\ref{b3203})
which are pseudo-Riemannian metrics on a world manifold $X$
\cite{iva,sardz,sard11,higgs14a}.

An underlying physical reason for Equivalence Principle is the
existence of Dirac spinor fields which possess Lorentz spin
symmetries, but do not admit general covariant transformations
\cite{sardz,sard98a,sard11}.

In classical field theory, Dirac spinor fields usually are
represented by sections of a spinor bundle on a world manifold $X$
whose typical fibre is a Dirac spinor space $\Psi_\mathrm{D}$ and
whose structure group is a Lorentz spin group Spin$(1,3)$.

Note that spinor representations of Lie algebras $so(m,n-m)$ of
pseudo-orthogonal Lie groups $SO(m,n-m)$, $n\geq 1$,
$m=0,1,\ldots,n$, were discovered by E. Cartan in 1913, when he
classified finite-dimensional representations of simple Lie
algebras \cite{cartan}. Though, there is a problem of spinor
representations of pseudo-orthogonal Lie groups $SO(m,n-m)$
themselves. Spinor representations are attributes of Spin groups
Spin$(m,n-m)$. Spin groups Spin$(m,n-m)$ are two-fold coverings
(\ref{104}) of pseudo-orthogonal groups $SO(m,n-m)$.

Spin groups Spin$(m,n-m)$ are defined as certain subgroups of real
Clifford algebras $\cC(m,n-m)$ (\ref{104a}). Moreover, spinor
representations of Spin groups in fact are the restriction of
spinor representation of Clifford algebras to its Spin subgroups.
Indeed, one needs an action of a whole Clifford algebra in a
spinor space in order to construct a Dirac operator. In 1935, R.
Brauer and H. Weyl described spinor representations in terms of
Clifford algebras \cite{brauer,law}.

Our approach to describing spinors is the following.

$\bullet$ We are based on the fact that real Clifford algebras
$\cC(m,n-m)$ and complex Clifford algebras $\mathbb C\cC(n)$ of
even dimension $n$ are isomorphic to matrix algebras (Theorems
\ref{k3} and \ref{k8}, respectively). Therefore, they are simple
(Corollaries \ref{k6} and \ref{k12}), and all their automorphisms
are inner (Theorems \ref{k15} and \ref{k21}). Their invertible
elements constitute general linear matrix groups (Theorems
\ref{sp501} and \ref{sp503}). They act on Clifford algebras by a
left-regular representation, and their adjoint representation as
projective linear groups exhaust all automorphisms of Clifford
algebras (Theorems \ref{k33} and \ref{sp504}).

Note that this just the case of a Clifford algebra $\cC(1,3)$ in
gravitation theory (Part II).

$\bullet$ Real and complex Clifford algebras of odd dimension $n$
are described as even subrings of Clifford algebras of even
dimension (Lemmas \ref{314} and \ref{sp507}, Example \ref{k2}).

$\bullet$ Given a real Clifford algebra $\cC(m,n-m)$, the
corresponding spinor space $\Psi(m,n-m)$ is defined as a carrier
space of its exact irreducible representation (Definition
\ref{sp510}). We are based on the fact that an exact irreducible
representation of a real Clifford algebra $\cC(m,n-m)$ of even
dimension $n$ is unique up to an equivalence, whereas a Clifford
algebra $\cC(m,n-m)$ of odd dimension $n$ admits two inequivalent
irreducible representations (Theorem \ref{a10}).

In particular, a Dirac spinor space is defined to be a spinor
space $\Psi(1,3)$ of a Clifford algebra $\cC(1,3)$ (Example
\ref{k1}).

However, Examples \ref{r88} -- \ref{r8} of Clifford algebras
$\cC(0,2)$ and $\cC(2,0)$, respectively, show that spinor spaces
$\Psi(m,n-m)$ and $\Psi(m',n-m')$ need not be isomorphic vector
spaces for $m'\neq m$. For instance, a Dirac spinor space
$\Psi(1,3)$ differs from a Majorana spinor space $\Psi(3,1)$ of a
Clifford algebra $\cC(3,1)$ (Example \ref{sp600}). In contrast
with the four-dimensional real matrix representation (\ref{62}) of
$\cC(3,1)$, a representation of a Clifford algebra $\cC(3,1)$ by
complex Dirac's matrices (\ref{41a}) is not a representation a
real Clifford algebra by virtue of Definition \ref{sp602}. Indeed,
from the physical viewpoint, Dirac spinor fields describing
charged fermions are complex fields.

$\bullet$  We therefore focus our consideration on complex
Clifford algebras and
 complex spinors. A complex Clifford algebra $\mathbb
C\cC(n)$ (Definition \ref{ss19}) of even dimension $n$ is
isomorphic to a ring $\mathrm{Mat}(2^{n/2}, \mathbb C)$ of complex
$(2^{n/2}\times 2^{n/2})$-matrices (Theorem \ref{k8}). Its
invertible elements constitute a general linear matrix group
$GL(2^{n/2}, \mathbb C)$ whose adjoint representation in $\mathbb
C\cC(n)$ yields the projective linear group $PGL(2^{n/2}, \mathbb
C)$ (\ref{k25}) of automorphisms of $\mathbb C\cC(n)$ (Theorem
\ref{k21}).

$\bullet$ Given a complex Clifford algebra $\mathbb C\cC(n)$, the
corresponding complex spinor space $\Psi(n)$ is defined as a
carrier space of its exact irreducible representation (Definition
\ref{ss34}). Similarly to a case of real Clifford algebras, we are
based on the fact that an exact irreducible representation of a
complex Clifford algebra $\mathbb C\cC(n)$ of even dimension $n$
is unique up to an equivalence, whereas a Clifford algebra
$\mathbb C\cC(n)$ of odd dimension $n$ admits two inequivalent
irreducible representations (Theorem \ref{a11}). Due to the
canonical monomorphism $\cC(m,n-m)\to \mathbb C\cC(n)$
(\ref{sp200}) of real Clifford algebras to the complex ones, a
complex spinor space $\Psi(n)$ admits a representation of a real
Clifford algebra $\cC(m,n-m)$, though it need not be irreducible.

$\bullet$ In accordance with Definition \ref{ss34} and Theorem
\ref{a11}, we define a particular complex Clifford space $\Psi(n)$
in a case of even $n$ as a minimal left ideal of a complex
Clifford algebra $\mathbb C\cC(n)$ (Definition \ref{ss33}). Thus,
a spinor representation
\mar{sp110}\beq
\gamma:\mathbb C\cC(n)\times \Psi(n) \to \Psi(n) \label{sp110}
\eeq
of a Clifford algebra $\mathbb C\cC(n)$ is equivalent to the
canonical representation of $\mathrm{Mat}(2^{n/2}, \mathbb C)$ by
matrices in a complex vector space $\Psi(n)=\mathbb C^{2^{n/2}}$
(Theorem \ref{k53}). Moreover, this spinor space $\Psi(n)$ also
carries out a left-regular irreducible representation of a general
linear matrix group $GL(2^{n/2},\mathbb C)=\cG\mathbb C\cC(n)$
which is equivalent to the natural matrix representation of
$GL(2^{n/2},\mathbb C)$ in $\mathbb C^{2^{n/2}}$ (Corollary
\ref{k54}). Thus, this group preserves spinor spaces.

$\bullet$ We show that any complex spinor space $\Psi(n)$ as a
minimal left ideal is generated by some Hermitian idempotent $p\in
\Psi(n)$ (\ref{a21}) (Theorem \ref{a41}), and obtain its group of
automorphisms. A key point, that a spinor subspace $\Psi(n)$ of a
complex Clifford algebra $\mathbb C\cC(n)$ is not unique, and it
is not stable under automorphisms of $\mathbb C\cC(n)$.

Treating a complex spinor space $\Psi(n)$ as a subspace (i.e. a
minimal left ideal) of a complex Clifford algebra $\mathbb
C\cC(n)$ which carries out its left-regular representation
(\ref{sp110}), we believe reasonable to consider a fibre bundle in
spinor spaces $\Psi(n)$ as a subbundle of a fibre bundle in
Clifford algebras. However, one usually considers fibre bundles in
Clifford algebras whose structure group is a group of
automorphisms of these algebras \cite{book09,law} (Remark
\ref{sp103}). A problem is that, as was mentioned above, this
group fails to preserve spinor subspaces $\Psi(n)$ of a Clifford
algebra $\mathbb C\cC(n)$ and, thus, it can not be a structure
group of spinor bundles.

Therefore, we define fibre bundles $\lC$ (\ref{sp100}) in Clifford
algebras $\mathbb C\cC(n)$ whose structure group is a general
linear group $GL(2^{n/2}, \mathbb C)$ of invertible elements of
$\mathbb C\cC(n)$ which acts on this algebra by left
multiplications (Definition \ref{ps101}). Certainly, it preserves
minimal left ideals of this algebra and, consequently, is a
structure group of spinor subbundles $S$ of a Clifford algebra
bundle $\lC$ (Definition \ref{sp106}).

In particular, let $X$ be a smooth real manifold of dimension
$2^{n/2}$, $n=2,4,\ldots$. Let $TX$ be the tangent bundle over
$X$. Their structure group is $GL(2^{n/2}, \mathbb R)$. Due to the
canonical group monomorphism $GL(2^{n/2}, \mathbb R) \to
GL(2^{n/2}, \mathbb C)$ \ref{sp211}, the complexification
${\mathbb C}TX$ (\ref{sp212}) of $TX$ can be represented as a
spinor bundle (Remark \ref{sp210}). This bundle admits general
covariant transformations and, thus, it is a natural bundle.

It should be emphasized that, though there is the ring
monomorphism $\cC(m,n-m)\to \mathbb C\cC(n)$ (\ref{sp200}), the
Clifford algebra bundle $\lC$ (\ref{sp100}) need not contain a
subbundle in real Clifford algebras $\cC(m,n-m)$ unless a
structure group $GL(2^{n/2}, \mathbb C)$ of $\lC$ is reducible to
a group $\cG\cC(m,n-m)$ of invertible elements of $\cC(m,n-m)$. We
study this condition (Section 6.1). Let $X$ be an $n$-dimensional
smooth manifold and $LX$ a principal frame bundle over $X$. We
show that any global section $h$ of the quotient bundle
$\Si(m,n-m)=LX/O(m,n-m)\to X$ (\ref{sp226}) is associated to the
fibre bundle $\lC^h\to X$ (\ref{sp240}) in complex Clifford
algebras $\mathbb C\cC(n)$ which contains the subbundle
$\lC^h(m,n-m)\to X$ (\ref{sp241}) in real Clifford algebras
$\cC(m,n-m)$ and a spinor subbundle $S^h\to X$.

A key point is that, given different sections $h$ and $h'$ of the
quotient bundle $\Si(m,n-m)\to X$ (\ref{sp226}), the Clifford
algebra bundles $\lC^h$ and $\lC^{h'}$ need not be isomorphic. In
order to describe all these non-isomorphic Clifford algebra
bundles $\lC^h$, follow a construction of composite bundles
(Section 6.2). We consider composite Clifford algebra bundles
$\lC_\Si$ (\ref{sp246}) and $\lC(m,n-m)_\Si$ (\ref{sp247}), and
the spinor bundle $S_\Si$ (\ref{sp248}) over a base $\Si(m,n-m)$
(\ref{sp226}. Then given a global section $h$ of the quotient
bundle $\Si(m,n-m)\to X$ (\ref{sp226}), the pull-back bundles
$h^*\lC_\Si$, $h^*\lC(m,n-m)_\Si$ and $h^*S_\Si$ are the above
mentioned fibre bundles $\lC^h\to X$, $\lC^h(m,n-m)\to X$ and
$S^h\to X$, respectively.

This is just the case of gravitation theory where, in order to
define a Dirac operator, we must consider a fibre bundle in
Clifford algebras $\cC(1,3)$ whose generating spaces are cotangent
spaces to a world manifold $X$.

In forthcoming Part II of our work, following the technique of
composite Clifford algebra bundles in Section 6.2, we consider
composite Clifford algebra bundles $\lC_\Si$ and $\lC(1,3)_\Si$,
and a spinor bundle $S_\Si$ over the base $LX/SO(1,3)$
(\ref{b3203}). As was mentioned above, global sections $h$ of the
quotient bundle $LX/SO(1,3)\to X$ are pseudo-Riemannian metrics on
$X$. Given such a section, the corresponding pull-back bundles
$h^*\lC_\Si$, $h^*\lC(1,3)_\Si$ and $h^*S_\Si$ are $h$-associated
fibre bundles $\lC^h\to X$, $\lC^h(1,3)\to X$ and $S^h\to X$ over
$X$, respectively.

\section{Clifford algebras}

A real Clifford algebra is defined as a ring (i.e., a unital
associative algebra) possessing a certain vector subspace of
generating elements (Definition \ref{ss15}). However, such a ring
can possess different generating spaces. Therefore, we also
consider a real Clifford algebra without specifying its generating
space. Complex Clifford algebras are defined as the
complexification of the real ones (Definition \ref{ss19}).

\subsection{Real Clifford algebras}

Let $V=\mathbb R^n$ be an $n$-dimensional real vector space
provided with a non-degenerate bilinear form (a pseudo-Euclidean
metric) $\eta$. Let us consider a tensor algebra
\be
\otimes V= \mathbb{R} \oplus V\oplus \op\otimes^2V\oplus\cdots
\oplus \op\otimes^k V\oplus \cdots
\ee
of $V$ and its two-sided ideal $I_\eta$ generated by the elements
\be
v\otimes v'+v'\otimes v - 2\eta(v,v')e, \qquad  v,v'\in V,
\ee
where $e$ denotes the unit element of $\otimes V$. The quotient
$\otimes V/I_\eta$ is a real non-commutative ring.

\begin{definition} \label{ss15} \mar{ss15}
A real ring $\otimes V/I_\eta$ together with a fixed generating
subspace $(V,\eta)$ is called the real Clifford algebra
$\cC(V,\eta)$ modelled over a pseudo-Euclidean space $(V,\eta)$.
\end{definition}

\begin{remark} \label{ss17} \mar{ss17}
Unless otherwise stated, by a Clifford algebra hereafter is meant
a real Clifford algebra in Definition \ref{ss15}.
\end{remark}

There is the canonical monomorphism of a real vector space $V$ to
the quotient  $\otimes V/I_\eta$. It is a generating subspace of a
real ring $\otimes V/I_\eta$. Its elements obey the relations
\be
vv'+v'v - 2\eta(v,v')e=0, \qquad  v,v'\in V.
\ee

\begin{definition} \label{ss16} \mar{ss16}
Given Clifford algebras $\cC(V,\eta)$ and $\cC(V',\eta')$, by
their isomorphism is meant an isomorphism of them as rings:
\mar{ss2}\beq
\phi:\cC(V,\eta) \to \cC(V',\eta'), \qquad
\phi(qq')=\phi(q)\phi(q'), \label{ss2}
\eeq
which also is an isometric isomorphism of their generating
pseudo-Euclidean spaces:
\mar{ss1}\ben
&& \phi:\cC(V,\eta)\supset (V,\eta)\to
(V',\eta')\subset\cC(V',\eta'), \label{ss1} \\
&& 2\eta'(\phi(v),\phi(v'))=\phi(v)\phi(v') + \phi(v')\phi(v)=\phi(vv' + v'v)
=2\eta(v,v').\nonumber
\een
\end{definition}

It follows from the isomorphism (\ref{ss1}) that two Clifford
algebras $\cC(V,\eta)$ and $\cC(V',\eta')$ are isomorphic iff they
are modelled over pseudo-Euclidean spaces $(V,\eta)$ and
$(V',\eta')$ of the same signature. Let a pseudo-Euclidean metric
$\eta$ be of signature $(m;n-m)=(1,...,1;-1,...,-1)$. Let
$\{v^1,...,v^n\}$ be a basis for $V$ such that $\eta$ takes a
diagonal form
\be
\eta^{ab}=\eta(v^a,v^b)=\pm \delta^{ab}.
\ee
Then a ring $\cC(V,\eta)$ is generated by elements $v^1,...,v^n$
which obey relations
\be
v^a v^b+v^b v^a=2\eta^{ab}e.
\ee
We agree to call $\{v^1,...,v^n\}$ the basis for a Clifford
algebra $\cC(\mathbb R^n,\eta)$. Given this basis, let us denote
$\cC(\mathbb R^n,\eta)=\cC(m,n-m)$.

In accordance with Definition \ref{ss16}, any isomorphism
(\ref{ss2}) -- (\ref{ss1}) of Clifford algebras is their ring
isomorphism (\ref{ss2}). However, the converse is not true,
because their ring isomorphism (\ref{ss2}) need not be the
isometric isomorphism (\ref{ss1}) of their generating spaces.
Therefore, we also consider Clifford algebras, without specifying
their generating spaces.

\begin{lemma} \label{ss4} \mar{ss4}
Any isometric isomorphism (\ref{ss1}) of generating vector spaces
$\phi: V\to V'$ of Clifford algebras $\cC(V,\eta)$ and
$\cC(V',\eta')$ is prolonged to their ring isomorphism
(\ref{ss2}):
\mar{ss3}\beq
\phi:\cC(V,\eta)\to \cC(V',\eta') \qquad \phi(v_1\cdots v_k)=
\phi(v_1)\cdots \phi(v_k), \label{ss3}
\eeq
which also is an isomorphism of Clifford algebras.
\end{lemma}

\begin{remark} \mar{464} \label{464} Let $g$ be a general (non-isometric) linear automorphism of a generating
vector space $V$ of a Clifford algebra $\cC(V,\eta)$. It yields an
automorphism of $\cC(V,\eta)$ as a real vector space, but not its
ring automorphism because
\be
&& g(v)g(v')+ g(v')g(v)=2\eta(g(v),g(v'))e\neq \\
&& \qquad 2\eta(v,v')e=g(vv'+ v'v), \qquad v,v'\in V,
\ee
in general. Let us provide a vector space $V$ with a different
pseudo-Euclidean metric $\eta'$ such that
\be
\eta'(g(v),g(v'))= \eta(v,v'), \qquad v,v'\in V.
\ee
It is of the same signature as $\eta$. Then a morphism $g$ is an
isometric isomorphism of a pseudo-Euclidean space $(V,\eta)$ to a
pseudo-Euclidean space $(V,\eta')$. Accordingly it yields an
isomorphism of a Clifford algebra $\cC(V,\eta)$ to a Clifford
algebra $\cC(V,\eta')$ modelled over $(V,\eta')$.
\end{remark}

\begin{example} \mar{r5} \label{r5}
There are the following isomorphisms of real rings \cite{law}:
\mar{4,5}\ben
&& \cC(1,0)= \mathbb R\oplus
\mathbb R, \label{4}\\
&& \cC(0,1)= \mathbb C, \label{5}
\een
Let $\{r^1=1,r^2=1\}$ be a basis for a ring $\mathbb
R\oplus\mathbb R$. Then the isomorphism (\ref{4}) reads $(e\lrr
r^1)$, $v^1\lrr r^2$. Accordingly, the isomorphism (\ref{5}) takes
a form $e\lrr 1$, $v^1\lrr i$.
\end{example}

\begin{example} \mar{r88} \label{r88} There is a ring isomorphism
\mar{7}\beq
\cC(0,2)= \mathbb H, \label{7}
\eeq
where $\mathbb{H}$ is a real division ring of quaternions. An
underlying real vector space of $\mathbb{H}$ has a basis
$\{\mathbf 1,\tau^1,\tau^2,\tau^3\}$ whose elements obey the
relations
\be
(\tau^1)^2=(\tau^2)^2=(\tau^3)^2=\tau^1\tau^2\tau^3=-\mathbf 1,
\ee
where $\mathbf 1$ is the unit element of $\mathbb{H}$. These
relations define the real division ring $\mathbb{H}$ with two
generating elements, e.g., $\tau^1$ and $\tau^2$. We have
\be
\tau^1\tau^2=-\tau^2\tau^1=\tau^3, \qquad
\tau^2\tau^3=-\tau^3\tau^2=\tau^1, \qquad
\tau^3\tau^1=-\tau^1\tau^3=\tau^2.
\ee
Due to an isomorphism
\mar{80}\beq
\mathbb C\op\otimes_{\mathbb R}\mathbb H = \mathrm{Mat}(2,\mathbb
C), \label{80}
\eeq
a quaternion division ring $\mathbb{H}$ can be represented as a
real subalgebra of an algebra $\mathrm{Mat}(2,\mathbb C)$ of
complex $(2\times 2)$-matrices whose underlying real vector space
possesses a basis
\mar{51}\beq
\mathbf 1=\begin{pmatrix}1 & 0\\0 & 1
\end{pmatrix}, \quad \tau^1=\begin{pmatrix}0 & -i\\-i & 0
\end{pmatrix}, \quad \tau^2=\begin{pmatrix}0 & -1\\1 &
0\end{pmatrix}, \quad \tau^3=\begin{pmatrix}-i & 0\\0 &
i\end{pmatrix}, \label{51}
\eeq
so that $\tau^k=-i\sigma^k$, $k=1,2,3$, where $\sigma^k$ are the
Pauli matrices
\mar{52}\beq
\sigma^1=\begin{pmatrix}0 & 1\\1 & 0 \end{pmatrix}, \qquad
\sigma^2=\begin{pmatrix}0 & -i\\i & 0\end{pmatrix}, \qquad
\sigma^3=\begin{pmatrix}1 & 0\\0 & -1\end{pmatrix}. \label{52}
\eeq
Then the isomorphism $\cC(0,2)= {\mathbb{H}}$ (\ref{7}) can be
written in a form $v^1\lrr \tau^1$, $v^2\lrr \tau^2$, but it is
not canonical. Using this isomorphism and the matrix
representation (\ref{51}) of a quaternion division ring
$\mathbb{H}$, we obtain a matrix representation of a Clifford
algebra $\cC(0,2)$ as a real subalgebra of an algebra
$\mathrm{Mat}(2,\mathbb C)$. Its underlying real vector space
possesses a basis
\mar{57}\beq
e=\mathbf 1, \qquad v^k=\tau^k, \qquad k=1,2,3. \label{57}
\eeq
\end{example}

It may happen that a ring $\cC(V,\eta)$ admits a generating
pseudo-Euclidean space $(V',\eta')$ whose signature differs from
that of $(V,\eta)$. In this case, $\cC(V,\eta)$ possesses the
structure of a Clifford algebra $\cC(V',\eta')$ which is not
isomorphic to a Clifford algebra $\cC(V,\eta)$.

\begin{lemma} \mar{303} \label{303} There are ring
isomorphisms
\mar{267,13}\ben
&&\cC(m,n-m)=\cC(n-m+1,m-1),\label{267}\\
&&\cC(m,n-m)=\cC(m-4,n-m+4), \qquad n,m\geq 4.\label{13}
\een
\end{lemma}

\begin{proof}
Let us consider a Clifford algebra $\cC(m,n-m)$ of $m>0,n>1$,
possessing a basis $\{v^1,...,v^n\}$. A ring $\cC(m,n-m)$ also is
generated by elements
\mar{266}\beq
w^1=v^1, \qquad w^i=v^1v^i, \qquad i>1. \label{266}
\eeq
These elements obey the relations
\be
&& w^i w^k+ w^k w^i =2\eta'^{ik}e, \\
&& \eta'^{11}=\eta^{11}, \qquad \eta'^{1k}=0,\qquad
\eta'^{ik}=-\eta^{11}\eta^{ik}, \qquad i,k>1.
\ee
Hence, a ring $\cC(m,n-m)$ also is a Clifford algebra modelled
over a pseudo-Euclidean space $(\mathbb R^n,\eta')$ of signature
$(1+n-m;m-1)$. Thus, we have the ring isomorphism (\ref{267})
given by the relations (\ref{266}). Turn now to the isomorphism
(\ref{13}). Let $(v^0,v^1,v^2,v^3, v^i)$ and $(w^0,w^1,w^2,w^3,
w^i)$ be bases for Clifford algebras $\cC(m,n-m)$ and
$\cC(m-4,n-m+4)$, respectively. Then their isomorphism (\ref{13})
is given by identifications $w^i\lrr v^i$ and
\mar{307}\beq
w^0\lrr v^1v^2v^3, \quad w^1\lrr v^0v^2v^3, \quad w^2\lrr
v^0v^1v^3, \quad w^3\lrr v^0v^1v^2. \label{307}
\eeq
\end{proof}

\begin{example} \mar{r8} \label{r8}
We have a ring isomorphism
\mar{6}\beq
\cC(2,0)= \cC(1,1)= \mathrm{Mat}(2,\mathbb R), \label{6}
\eeq
where $\mathrm{Mat}(2, \mathbb R)$ is a real ring of $(2\times
2)$-matrices. It exemplifies the isomorphism (\ref{267}) which
reads: $w^1\lrr e^1$, $w^2\to e^1e^2$, where $\{e^1,e^2\}$ and
$\{w^1,w^2\}$ are the bases for $\cC(2,0)$ and $\cC(1,1)$,
respectively. The matrix representation (\ref{6}) of $\cC(2,0)$ by
$\mathrm{Mat}(2,\mathbb R)$ takes a form
\mar{60}\beq
e^1=\sigma^1=
\begin{pmatrix}0 & 1\\ 1 & 0\end{pmatrix}, \qquad e^2=\sigma^3
= \begin{pmatrix}1 & 0\\0 & -1\end{pmatrix}. \label{60}
\eeq
Accordingly, the matrix representation $\cC(1,1)$ by
$\mathrm{Mat}(2,\mathbb R)$ is
\mar{60'}\beq
w^1=\sigma^1=
\begin{pmatrix}0 & 1\\ 1 & 0\end{pmatrix}, \qquad w^2=\tau^2
= \begin{pmatrix} 0& -1\\1 & 0\end{pmatrix}. \label{60'}
\eeq
The real rings $\cC(2,0)$ (\ref{60}) and $\cC(1,1)$ (\ref{60'})
coincide with each other. Their underlying vector space in
$\mathrm{Mat}(2,\mathbb R)$ possesses a basis $\{\mathbf 1,
\sigma^1, \sigma^3, \tau^2\}$.
\end{example}

With the real ring isomorphism (\ref{6}), we obtain the following
recursion relation.

\begin{lemma} \label{305} \mar{305} There is a real ring isomorphism
\mar{8}\beq
\cC(p+1,q+1) = \cC(1,1)\otimes \cC(p,q)
=\mathrm{Mat}(2,\cC(p,q)),\label{8}
\eeq
where $\mathrm{Mat}(2,\cC(p,q))$ is an algebra of $2\times 2$
matrices with entries in $\cC(p,q)$.
\end{lemma}

\begin{proof}
The isomorphisms (\ref{8}) take a form
\mar{55}\ben
&& v^+=w^1\otimes e=\begin{pmatrix}0 & \mathbf 1\\ \mathbf 1 &
0\end{pmatrix},  \qquad v^-=w^2\otimes e=\begin{pmatrix}0 & -\mathbf 1\\
\mathbf 1 & 0
\end{pmatrix}, \label{55}\\
&& v^i=w^1w^2\otimes e^i=\begin{pmatrix}\tau^i & 0\\0 &
-\tau^i\end{pmatrix}, \nonumber
\een
where $\{v^i, v^+, v^-\}$ is a basis for $\cC(p+1,q+1)$, $\{e^i\}$
is that for $\cC(p,q)$ and $\{w^1,w^2\}$ is the basis (\ref{60'})
for $\cC(1,1)$.
\end{proof}

\begin{example} \label{k1} \mar{k1}
Using isomorphisms (\ref{7}), (\ref{267}), (\ref{13}), (\ref{6})
and (\ref{8}), one can obtain the real ring isomorphisms
\mar{11}\beq
\cC(1,3)= \cC(4,0)= \cC(0,4)=  \cC(1,1)\otimes \cC(0,2) =
\mathrm{Mat}(2,\mathbb H), \label{11}
\eeq
The isomorphism $\cC(4,0)= \cC(0,4)$ (\ref{11}) exemplifies the
isomorphism (\ref{13}) given by the identification (\ref{307}). In
view of the formulas (\ref{57}) and (\ref{55}), the matrix
representation $\cC(1,3)= \mathrm{Mat}(2,\mathbb H)$ (\ref{11})
reads
\mar{41}\beq
v^+=\begin{pmatrix} 0 & \mathbf{1} \\ \mathbf{1} & 0
\end{pmatrix}, \qquad
v^-= \begin{pmatrix} 0 & -\mathbf{1} \\ \mathbf{1} & 0
\end{pmatrix}, \qquad v^{1,2}= \begin{pmatrix} -i\sigma^{1,2} & 0
\\ 0 & i\sigma^{1,2} \end{pmatrix},
\label{41}
\eeq
where $\sigma^{1,2}$ are the Pauli matrices (\ref{52}). Let us
call it the standard representation, though it is not canonical.
In particular, one usually deal with a representation of
$\cC(1,3)$ by Dirac's matrices
\mar{41a}\beq
 v^0=\gamma^0=\begin{pmatrix}\mathbf{1} & 0 \\ 0 & -\mathbf{1}
\end{pmatrix}, \qquad v^j=\gamma^j= \begin{pmatrix}0 &
-\sigma^j\\\sigma^j & 0 \end{pmatrix} . \label{41a}
\eeq
Let us also mention its different representation by other Dirac's
matrices
\mar{41'}\ben
&& \wt\gamma^\mu= S\gamma^\mu S^{-1}, \qquad
S=\frac{1}{\sqrt 2}\begin{pmatrix}\mathbf{1} & -\mathbf{1} \\
\mathbf{1} & \mathbf{1} \end{pmatrix}, \nonumber\\
&&\wt\gamma^0=\begin{pmatrix}0 & \mathbf{1} \\ \mathbf{1} & 0
\end{pmatrix}, \qquad \wt\gamma^j= \begin{pmatrix}0 &
-\sigma^j\\\sigma^j & 0 \end{pmatrix} . \label{41'}
\een
The isomorphism $\cC(4,0)= \cC(1,3)$ (\ref{11}) exemplifies the
isomorphisms (\ref{267}). Given the matrix representation
(\ref{41'}) of $\cC(1,3)$, it provides the matrix representation
\mar{250}\beq
w^0=\begin{pmatrix}0 & \mathbf{1} \\ \mathbf{1} & 0
\end{pmatrix}, \qquad w^j= \begin{pmatrix} \sigma^j & 0\\
0 & -\sigma^j \end{pmatrix} \label{250}
\eeq
of a Clifford algebra $\cC(4,0)$
\end{example}

\begin{example} \label{sp600} \mar{sp600}
Using isomorphisms (\ref{267}), (\ref{6}) and (\ref{8}), one can
obtain the real ring isomorphisms
\mar{12}\beq
\cC(2,2)= \cC(3,1)= \cC(1,1)\otimes \cC(0,2)=
\mathrm{Mat}(4,\mathbb R). \label{12}
\eeq
The formulas (\ref{60}) and (\ref{8}) lead to the representation
(\ref{12}) of $\cC(3,1)$ by real matrices:
\mar{62}\ben
&& v^+=\begin{pmatrix} 0 & \mathbf{1} \\ \mathbf{1} & 0
\end{pmatrix}, \qquad
v^-= \begin{pmatrix} 0 & -\mathbf{1} \\ \mathbf{1} & 0
\end{pmatrix}. \label{62}\\
&& v^1= \begin{pmatrix} \sigma^1 & 0
\\ 0 & -\sigma^1 \end{pmatrix} , \qquad v^2= \begin{pmatrix} \sigma^3 & 0
\\ 0 & -\sigma^3 \end{pmatrix} , \nonumber
\een
It is an irreducible four-dimensional representation of a Clifford
algebra $\cC(3,1)$. By virtue of Theorem \ref{a10}, this
irreducible representation is unique up to an equivalence.
\end{example}

Let $\cC^0(m,n-m)$ be a vector subspace of elements of a Clifford
algebra $\cC(m,n-m)$ which is spanned by polynomials in elements
of $\mathbb R^n$ of even degree. It is obviously a subring of a
ring $\cC(m,n-m)$, called its even subring.

\begin{lemma} \mar{314} \label{314}
There exists a ring isomorphism
\mar{100}\beq
\cC^0(m,n-m)= \cC(n-m, m-1), \qquad n>1. \label{100}
\eeq
\end{lemma}

\begin{proof} Let $\{v^0,\ldots,v^{n-1}\}$ and $\{w^1,\ldots,w^{n-1}\}$ be bases for
$\cC(m,n-m)$ and $\cC(m,n-m-1)$. Then the isomorphism (\ref{100})
is defined by the identification $w^i=v^0v^i$.
\end{proof}

\begin{example} \label{k2} \mar{k2}
Let us consider a Clifford algebra
\mar{311}\beq
\cC(0,3) = \cC^0(4,0). \label{311}
\eeq
Let a Clifford algebra $\cC(4,0)$ be represented by the matrices
$(\wt\gamma^0,-i\wt\gamma^j)$ (\ref{41'}). Then a Clifford algebra
$\cC(0,3)$ is generated by matrices
\mar{361}\ben
&& (a_0\wt\gamma^0 +ia_i\wt\gamma^i)(b_0\wt\gamma^0
+ib_i\wt\gamma^i)=(a_0\mathbf{1}
+ia_i\wt\gamma^i\wt\gamma^0)(b_0\mathbf{1}
+ib_i\wt\gamma^0\wt\gamma^i)= \nonumber \\
&&
\begin{pmatrix} c_0 \mathbf{1}+c_i\tau^i& 0\\0& d_0 \mathbf{1}+
d_i\tau^i
\end{pmatrix},  \qquad c_\mu,d_\mu\in\mathbb R. \label{361}
\een
Thus, there is a real ring isomorphism
\mar{320}\beq
\cC(0,3)= \mathbb H\times \mathbb H. \label{320}
\eeq
\end{example}

The recursion relation (\ref{305}) and the ring isomorphisms
(\ref{4}), (\ref{5}), (\ref{7}), (\ref{6}) and (\ref{320}) enable
us to provide the matrix representation of any real Clifford
algebra as follows.

\begin{theorem} \label{k3} \mar{k3} Clifford algebras $\cC(p,q)$
as rings are isomorphic to the following matrix algebras.
\mar{k5}\ben
&& \cC(p,q) = \label{k5}\\
&& \left\{
\begin{array}{ll}
\mathrm{Mat}(2^{(p+q)/2},\mathbb R)=\op\ot^{(p+q)/2}_{\mathbb R} \mathrm{Mat}(2,\mathbb R) & p-q=0,2\mod 8 \\
\mathrm{Mat}(2^{(p+q-1)/2},\mathbb R)\oplus
\mathrm{Mat}(2^{(p+q-1)/2},\mathbb R)  & p-q=1\mod 8  \\
\mathrm{Mat}(2^{(p+q-1)/2},\mathbb C) & p-q=3,7\mod 8 \\
\mathrm{Mat}(2^{(p+q-2)/2},\mathbb H) & p-q=4,6\mod 8 \\
\mathrm{Mat}(2^{(p+q-3)/2},\mathbb H)\oplus
\mathrm{Mat}(2^{(p+q-3)/2},\mathbb H)  & p-q=5\mod 8
\end{array}\right.\nonumber
\een
\end{theorem}

\begin{proof}
Owing to the isomorphism (\ref{13}), a Clifford algebra $\cC(p,q)$
is isomorphic to a Clifford algebra $\cC(p-4k,q+4k)$, $k\in\mathbb
Z$, so that $p-q-8k<8$. The we have eight different algebras
\be
\begin{array}{ll}
\cC((p+q)/2,(p+q)/2) & p-q=0,2\mod 8 \\
\cC((p+q+1)/2,(p+q-1)/2) & p-q=1\mod 8 \\
\cC((p+q+3)/2,(p+q-3)/2) & p-q=3,7\mod 8 \\
\cC((p+q-2)/2,(p+q+2)/2) & p-q=4,6\mod 8 \\
\cC((p+q-3)/2,(p+q+3)/2) & p-q=1 \mod 8.
\end{array}
\ee
Then the relations (\ref{305}) leads to the isomorphisms
\be
\begin{array}{ll}
\mathrm{Mat}(2^{(p+q)/2},\mathbb R)=\op\ot^{(p+q)/2}_{\mathbb R} \mathrm{Mat}(2,\mathbb R) & p-q=0,2\mod 8 \\
\mathrm{Mat}(2^{(p+q-1)/2},\cC(1,0))  & p-q=1\mod 8  \\
\mathrm{Mat}(2^{(p+q-1)/2},\cC(0,1)) & p-q=3,7\mod 8 \\
\mathrm{Mat}(2^{(p+q-2)/2},\cC(0,2)) & p-q=4,6\mod 8 \\
\mathrm{Mat}(2^{(p+q-3)/2},\cC(0,3)) & p-q=5\mod 8
\end{array}
\ee
The result (\ref{k5}) follows from the isomorphisms (\ref{4}),
(\ref{5}), (\ref{7}), (\ref{6}) and (\ref{320}).
\end{proof}

\begin{corollary} \label{k6} \mar{k6}
Since matrix algebras $\mathrm{Mat}(r,\cK)$, $\cK=\mathbb
R,\mathbb C, \mathbb H$, are simple, a glance at Table \ref{k5}
shows that real Clifford algebras $\cC(V,\eta)$ modelled over even
dimensional vector spaces $V$ (i.e., $p-q$ is even) are simple.
\end{corollary}

\begin{definition} \label{sp602} \mar{sp602}
By a representation of a Clifford algebra $\cC(V,\eta)$ is meant
its ring homomorphism $\rho$ to a real ring of linear
endomorphisms of a finite-dimensional real vector space $\Xi$,
whose dimension is called the dimension of a representation.
\end{definition}

For instance, the real matrix representation (\ref{62}) of a real
Clifford algebra $\cC(3,1)$ is its representation in accordance
with Definition \ref{sp602}. At the same time, a representation of
a Clifford algebra $\cC(3,1)$ by Dirac's matrices (\ref{41a}) is
not that by Definition \ref{sp602}.

A representation is said to be exact if $\rho$ is an isomorphism.
A representation is called irreducible if there is no proper
subspace of $\Xi$ which is a carrier space of a representation of
$\cC(V,\eta)$.

Two representations $\rho$ and $\rho'$ of a Clifford algebra
$\cC(V,\eta)$ in vector spaces $\Xi$ and $\Xi'$ are said to be
equivalent if there is an isomorphism $\xi: \Xi\to \Xi'$ of these
vector spaces such that $\rho'=\xi\circ\rho\circ\xi^{-1}$ is a
real ring isomorphism of $\rho(\cC(V,\eta))$ and
$\rho'(\cC(V,\eta))$.

The following is a corollary of Theorem \ref{k3}.

\begin{theorem} \label{a10} \mar{a10}
If $n=\mathrm{dim}\,V$ is even, an exact irreducible
representation of a real ring $\cC(m,n-m)$ is unique up to an
equivalence \cite{law}. If $n$ is odd there exist two inequivalent
exact irreducible representations of a Clifford algebra
$\cC(m,n-m)$.
\end{theorem}

\subsection{Complex Clifford algebras}

Let us consider the complexification
\mar{63}\beq
\mathbb C\cC(m,n-m)=\mathbb C\op\otimes_{\mathbb R}\cC(m,n-m)
\label{63}
\eeq
of a real ring $\cC(m,n-m)$. It is readily observed that all
complexifications $\mathbb C\cC(m,n-m)$, $m=0,\ldots, n$, are
isomorphic:
\mar{66}\beq
\mathbb C\cC(m,n-m)= \mathbb C\cC(m',n-m'), \label{66}
\eeq
both as real and complex rings. Namely, with the bases $\{v^i\}$
and $\{e^i\}$ for $\cC(m,n-m)$ and $\cC(n,0)$, their isomorphisms
(\ref{66}) are given by associations
\mar{65}\beq
v^{1,\ldots,m}\to e^{1,\ldots,m}, \qquad v^{m+1,\ldots,n}\to
ie^{m+1,\ldots,n}. \label{65}
\eeq

Though the isomorphisms (\ref{65}) are not unique, one can speak
about an abstract complex ring $\mathbb C\cC(n)$ (\ref{66}) so
that, given a real Clifford algebra $\cC(m,n-m)$ and its
complexification $\mathbb C\cC(m,n-m)$ (\ref{63}), there exist the
complex ring isomorphism (\ref{65}) of $\mathbb C\cC(m,n-m)$ to
$\mathbb C\cC(n)$.

\begin{definition} \label{ss19} \mar{ss19}
We call $\mathbb C\cC(n)$ (\ref{66}) the complex Clifford algebra,
and define it as a complex ring
\mar{a25}\beq
\mathbb C\cC(n)=\mathbb C\op\otimes_{\mathbb R}\cC(n,0),
\label{a25}
\eeq
generated by $n$ elements $(e^i)$ such that
\mar{210}\beq
e^ie^j+e^je^i==2\kappa(e^i,e^j)e=2\delta^{ij}e. \label{210}
\eeq
\end{definition}

Let us call $\{e^i\}$ (\ref{210}) the Euclidean basis for a
complex Clifford algebra $\mathbb C\cC(n)$. With this basis, any
element of $\mathbb C\cC(n)$ takes a form
\mar{a43}\beq
a=\lambda e +\op\sum_{1\leq k\leq n}\op\sum_{i_1<\ldots<
i_k}\lambda_{i_1\ldots i_k} e^{i_1}\cdots e^{i_k}, \qquad
\lambda,\lambda_{i_1\ldots i_k} \in \mathbb C. \label{a43}
\eeq

\begin{definition} \label{ss22} \mar{ss22}
A complex vector space $\cV$, spanned by an Euclidean basis
$\{e^i\}$ and provided with the bilinear form $\kappa$
(\ref{210}), is termed the Euclidean generating space of a complex
Clifford algebra $\mathbb C\cC(n)$.
\end{definition}

\begin{remark} \label{ss24} \mar{ss24}
Any generating space $(\cV,\kappa)$ of a complex Clifford algebra
is the Euclidean one with respect to some basis of $\cV$.
\end{remark}

\begin{lemma} \label{ss23} \mar{ss23}
The complex ring $\mathbb C\cC(n)$ (\ref{a25}) possesses a
canonical real subring
\mar{sp200}\beq
\cC(m,n-m)\to \mathbb C\cC(n) \label{sp200}
\eeq
with a basis
\mar{351}\beq
\{e^1,\ldots, e^m, ie^{m+1}, \ldots, ie^n\}. \label{351}
\eeq
\end{lemma}

\begin{remark} \label{sp640} \mar{sp640}
The definition (\ref{a25}) enables us to provide a complex
Clifford algebra $\mathbb C\cC(n)$ with an involution
\mar{a36}\beq
(\lambda e^i)^*=\ol\lambda e^i, \qquad (e^ie^j)^*=e^je^i, \qquad
\lambda\in \mathbb C, \label{a36}
\eeq
so that an involution of the element $a\in \mathbb C\cC(n)$
(\ref{a43}) reads
\mar{a33}\beq
a^*=\ol\lambda e +\op\sum_{1\leq k\leq n}\op\sum_{i_1<\ldots
i_k}\ol\lambda_{i_1\ldots i_k} e^{i_k}\cdots e^{i_1}. \label{a33}
\eeq
In particular, it follows that
\mar{a34}\beq
a^*a=\left(\ol\lambda\lambda + \op\sum_{1\leq k\leq
n}\op\sum_{i_1<\ldots <i_k}\ol\lambda_{i_1\ldots
i_k}\lambda_{i_1\ldots i_k}\right)e + \cdots \neq 0. \label{a34}
\eeq
An element $a\in \mathbb C\cC(n)$ is called the Hermitian one if
$a^*=a$. In this case, $a^2\neq 0$ in accordance with the formula
(\ref{a34}). The involution $*$ (\ref{a36}) makes a complex
Clifford algebra involutive. However, an automorphism of $\mathbb
C\cC(n)$ need not be its automorphism as an involutive algebra
(Remark \ref{ss30}).
\end{remark}

Theorem \ref{k3} provides the following classification of the
complex Clifford algebras $\mathbb C\cC(n)$ (\ref{a25}).

\begin{theorem} \label{k8} \mar{k8} Complex Clifford algebras are
isomorphic to the following matrix ones
\mar{k9}\beq
\mathbb C\cC(n) = \left\{
\begin{array}{ll}
\mathrm{Mat}(2^{n/2},\mathbb C)=\op\ot^{n/2}_{\mathbb C} \mathrm{Mat}(2,\mathbb C)
=\op\ot^{n/2}_{\mathbb C} \mathbb C\cC(2) & n=0\mod 2\\
\mathrm{Mat}(2^{(n-1)/2},\mathbb C)\oplus
\mathrm{Mat}(2^{(n-1)/2},\mathbb C)  & n=1\mod 2
\end{array}\right.\label{k9}
\eeq
\end{theorem}

\begin{corollary} \label{k12} \mar{k12}
Since matrix algebras $\mathrm{Mat}(n,\mathbb C)$ are simple and
central (i.e., their center is proportional to the unit matrix),
complex Clifford algebras $\mathbb C\cC(n)$ of even $n$ are
central simple algebras.
\end{corollary}

\begin{lemma} \label{sp507} \mar{sp507}
It follows from Definition \ref{ss19} and Lemma \ref{314}, that
complex Clifford algebra of odd dimension are even subrings of
Complex Clifford algebras of even dimension in Corollary
\ref{k12}.
\end{lemma}

\begin{example} \label{k10} \mar{k10} Let us consider a complex
Clifford algebra $\mathbb C\cC(2)$. There is its isomorphism
(\ref{k9}):
\mar{a26}\beq
\mathbb C\cC(2)=\mathrm{Mat}(2,\mathbb C). \label{a26}
\eeq
Its Euclidean basis in this representation is
\be
e^1= \rho^1=
\begin{pmatrix}0 & 1\\ 1 & 0\end{pmatrix}, \qquad e^2
= \rho^2=\begin{pmatrix}1 & 0\\0 & -1\end{pmatrix}.
\ee
Then its elements $M$ with respect  to this basis take a form
\be
\mathrm{Mat}(2,\mathbb C)\ni M= ae+a_1e^1 +a_2e^2 + be^1e^2,
\qquad a,a_1,a_2,b\in \mathbb C,
\ee
so that $M^*=M^+$ is a complex conjugate transposition of a matrix
$M\in \mathrm{Mat}(2,\mathbb C)$.
\end{example}

\begin{example} \label{k11} \mar{k11} Let us consider a complex
Clifford algebra $\mathbb C\cC(4)$. There is its isomorphism
(\ref{k9}):
\mar{a27}\beq
\mathbb C\cC(4)=\mathrm{Mat}(4,\mathbb C), \label{a27}
\eeq
such that $M^*=M^+$ is a complex conjugate transposition of a
matrix $M\in \mathrm{Mat}(4,\mathbb C)$. Let $\mathbb C\cC(4)$
(\ref{a27}) be generated by the elements (\ref{250}):
\mar{250'}\beq
\epsilon^0=\begin{pmatrix}0 & \mathbf{1} \\ \mathbf{1} & 0
\end{pmatrix}, \qquad \epsilon^j= \begin{pmatrix} \sigma^j & 0\\
0 & -\sigma^j \end{pmatrix} \label{250'}
\eeq
which obey the relations (\ref{210}). Let us introduce the
notation
\mar{431}\ben
&& \epsilon^{\alpha\beta}=\frac14(\epsilon^\alpha\epsilon^\beta
-\epsilon^\beta\epsilon^\alpha), \qquad
(\epsilon^{\alpha\beta})^2=-\frac14e, \nonumber \\
&& [\epsilon^{\alpha\beta},\epsilon^{\mu\nu}]=\delta^{\alpha\nu}\epsilon^{\beta\mu}
+
 \delta^{\beta\mu}\epsilon^{\alpha\nu} - \delta^{\alpha\mu}\epsilon^{\beta\nu} -
 \delta^{\beta\nu}\epsilon^{\alpha\mu}\nonumber\\
&& \epsilon^5= \epsilon^0\epsilon^1\epsilon^2\epsilon^3, \qquad
(\epsilon^5)^2=e, \qquad
\epsilon^\mu\epsilon^5=-\epsilon^5\epsilon^\mu,
\label{431}\\
&& \varepsilon^\mu=\epsilon^\mu\epsilon^5, \qquad
\varepsilon^\mu\varepsilon^\nu+ \varepsilon^\nu\varepsilon^\mu
=-2\delta^{\mu\nu}e.
\een
Then in accordance with the isomorphism (\ref{a27}), any element
of $M\in\mathrm{Mat}(4,\mathbb C)$ is represented by a sum
\mar{432}\beq
M=ae+a_\mu\epsilon^\mu + a_{\alpha\beta}\epsilon^{\alpha\beta} +
b_\mu\varepsilon^\mu +b\epsilon^5, \qquad a,a_\mu,a_{\alpha\beta},
b_\mu,b\in\mathbb C. \label{432}
\eeq
We also have the isomorphism (\ref{k9}):
\mar{223}\beq
\mathbb C\cC(4)=\mathbb C\cC(2)\op\otimes_{\mathbb C}\mathbb
C\cC(2). \label{223}
\eeq
Let $\{e^1,e^2\}$ be generating elements of a complex Clifford
algebra $\mathbb C\cC(2)$ which obeys the relations (\ref{210}).
Then for instance, one can choose the generating elements
\mar{230}\beq
\epsilon^0=e^1\otimes e, \quad \epsilon^1=ie^1e^2\otimes e^2,
 \quad \epsilon^2=ie^1e^2\otimes e^1,
\quad \epsilon^3=e^2\otimes e, \label{230}
\eeq
of a complex Clifford algebra $\mathbb C\cC(4)$. With the
generating elements (\ref{230}), the isomorphism (\ref{223}) takes
a form
\mar{433}\ben
&& \epsilon^{01}=\frac{i}2e^2\otimes e^2, \qquad \epsilon^{02}=\frac{i}2e^2\otimes
e^1, \qquad \epsilon^{03}=\frac12e^1e^2\otimes e, \nonumber\\
&& \epsilon^{12}=-\frac12e\otimes e^1e^2, \qquad \epsilon^{13}=\frac{i}2e^1\otimes
e^2, \qquad \epsilon^{23}=\frac{i}2e^1\otimes e^1, \label{433}\\
&& \varepsilon^0=e^2\otimes
e^1e^2, \quad \varepsilon^1=ie\otimes e^1, \quad
\varepsilon^2=-ie\otimes e^2, \quad \varepsilon^3=e^1\otimes
e^1e^2, \nonumber\\
&& \epsilon^5=-e^1e^2\otimes e^1e^2. \nonumber
\een
\end{example}

\begin{definition} \label{sp603} \mar{sp603}
By a representation of a complex Clifford algebra $\mathbb
C\cC(n)$ is meant its morphism $\rho$ to a complex algebra of
linear endomorphisms of a finite-dimensional complex vector space.
\end{definition}

The following is a Corollary of Theorem \ref{k8}.

\begin{theorem} \label{a11} \mar{a11}
If $n$ is even, an exact irreducible representation of a complex
Clifford algebra $\mathbb C\cC(n)$ is unique up to an equivalence
\cite{law}. If $n$ is odd there exist two inequivalent exact
irreducible representations of a complex Clifford algebra $\mathbb
C\cC(n)$.
\end{theorem}

\begin{remark} \label{sp608} \mar{sp608}
Throughout the work, by representations of real and complex
Clifford algebras are meant their exact representations only.
\end{remark}

In view of Corollary \ref{k12} and Theorem \ref{a11}, we hereafter
focus our consideration on real and complex Clifford algebras
modelled over even vector spaces, and describe Clifford algebras
of odd dimension as even subrings of those of even dimension
(Lemmas \ref{314} and and \ref{sp507}, Example \ref{k2}).

\section{Automorphisms of Clifford algebras}

We consider both generic ring automorphisms of a Clifford algebra
and its automorphisms which preserve a specified generating space.

\subsection{Automorphisms of real Clifford algebras}

Let $\cC(V,\eta)$ be a real Clifford algebra modelled over an
even-dimensional pseudo-Euclidean space $(V,\eta)$. By
Aut$[\cC(V,\eta)]$ is denoted the group of automorphisms of a real
ring $\cC(V,\eta)$. A key point is the following.

\begin{theorem} \label{k15} \mar{k15}
Any automorphism of a real ring $\cC(V,\eta)$ is inner.
\end{theorem}

\begin{proof}
Theorem \ref{k3} states that any real Clifford algebra $\cC(p,q)$,
$p-q=0\mod 2$ as a ring is isomorphic to some matrix algebra
$\mathrm{Mat}(m,\cK)$, $\cK=\mathbb R,\mathbb C,\mathbb H$. Such
an algebra is simple. Algebras $\mathrm{Mat}(m,\cK)$, $\cK=\mathbb
R,\mathbb H$, are central simple real algebras with the center
$\cZ=\mathbb R$. Algebras $\mathrm{Mat}(m,\mathbb C)$ are central
simple complex algebras with the center $\cZ=\mathbb C$. In
accordance with the well-known Skolem--Noether theorem
automorphisms of these algebras are inner.
\end{proof}

\begin{theorem} \label{sp501} \mar{sp501}
Invertible elements of a Clifford algebra
$\cC(V,\eta)=\mathrm{Mat}(m,\cK)$ constitute a general linear
matrix group
\mar{sp611}\beq
\cG\cC(V,\eta)=Gl(m,\cK). \label{sp611}
\eeq
\end{theorem}

In particular, this group contains all elements $v\in V\subset
\cC(V,\eta)$ such that $\eta(v,v)\neq 0$. Acting in $\cC(V,\eta)$
by left and right multiplications, the group $\cG\cC(V,\eta)$
(\ref{sp611}) also acts in a Clifford algebra by the adjoint
representation
\mar{a5}\beq
 \wh g: q\to gqg^{-1}, \qquad g\in \cG\cC(V,\eta),\qquad q\in
\cC(V,\eta). \label{a5}
\eeq
By virtue of Theorem \ref{k15}, this representation provides an
epimorphism
\mar{83}\beq
\zeta:\cG\cC(V,\eta)=Gl(m,\cK) \to
Gl(m,\cK)/\cZ=\mathrm{Aut}[\cC(V,\eta)]. \label{83}
\eeq
Thus, we come to the following.

\begin{theorem} \label{k33} \mar{k33}
The group of automorphisms of a real Clifford algebra
$\cC(V,\eta)=\mathrm{Mat}(m,\cK)$, $\cK=\mathbb R,\mathbb
C,\mathbb H$, is a projective linear group
\mar{k34}\beq
\mathrm{Aut}[\cC(V,\eta)]=PGl(m,\cK)=Gl(m,\cK)/\cZ, \label{k34}
\eeq
where $\cZ=\mathbb R$ if $\cK=\mathbb R, \mathbb H$ and
$\cZ=\mathbb C$ if $\cK=\mathbb C$.
\end{theorem}

Any ring automorphism $g$ of $\cC(V,\eta)$ sends a generating
pseudo-Euclidean space $(V,\eta)$ of $\cC(V,\eta)$ onto an
isometrically isomorphic pseudo-Euclidean space $(V',\eta')$ such
that
\be
2\eta'(g(v),g(v'))e=g(v)g(v')+g(v')g(v)= 2\eta(v,v')e, \qquad
v,v'\in V.
\ee
It also is a generating space of a ring $\cC(V,\eta)$. Conversely,
let $(V,\eta)$ and $(V',\eta')$ be two different pseudo-Euclidean
generating spaces of the same signature of a ring $\cC(V,\eta)$.
In accordance with Lemma \ref{ss4}, their isometric isomorphism
$(V,\eta) \to (V',\eta')$ gives rise to an automorphism of a ring
$\cC(V,\eta)$ which also is an isomorphism of Clifford algebras
$\cC(V,\eta)\to \cC(V',\eta')$.

In particular, any (isometric) automorphism
\be
 g:V\ni v\to
g(v)\in V, \qquad \eta(g(v),g(v'))=\eta(v,v'), \qquad g\in
O(V,\eta),
\ee
of a pseudo-Euclidean generating space $(V,\eta)$ is prolonged to
an automorphism of a ring $\cC(V,\eta)$ which also is an
automorphism of a Clifford algebra $\cC(V,\eta)$. Then we have a
monomorphism
\mar{82}\beq
O(V,\eta)\to \mathrm{Aut}[\cC(V,\eta)\,]. \label{82}
\eeq
of a group $O(V,\eta)$ of automorphisms of a pseudo-Euclidean
space $(V,\eta)$ to a group of ring automorphisms of
$\cC(V,\eta)$. Herewith, an automorphism $g\in O(V,\eta)$ of a
ring $\cC(V,\eta)$ is the identity one iff its restriction to $V$
is an identity map of $V$. Consequently, the following is true.

\begin{theorem} \label{a3} \mar{a3}
A subgroup $O(V,\eta)\subset \mathrm{Aut}[\cC(V,\eta)\,]$
(\ref{82}) exhausts all automorphisms of a ring $\cC(V,\eta)$
which are automorphisms of a Clifford algebra $\cC(V,\eta)$.
\end{theorem}

\begin{remark} \mar{L2} \label{L2}
Elements of $O(V,\eta)$ are represented by inner automorphisms of
$\cC(V,\eta)$ as follows. Given an element $w\in V$,
$\eta(w,w)\neq 0$, let
\be
w^\perp=\{v\in V; \, \, \eta(v,w)=0\}
\ee
be a hyperplane in $V$ which is pseudo-orthogonal to $w$ with
respect to a metric $\eta$. Then any element $v\in V$ is
represented by a sum
\be
v= u + \frac{\eta(v,w)}{\eta(w,w)}w, \qquad u\in w^\perp.
\ee
Let us consider the inner automorphism $\wh w$ (\ref{a5}). Its
restriction to $V$ reads
\mar{88}\ben
&&\wh w: V\ni v\to wvw^{-1} =-v+2\frac{\eta(w,v)}{\eta(w,w)}w\in V, \label{88}\\
&& \eta(wvw^{-1},wvw^{-1})=\eta(v,v). \nonumber
\een
It is an automorphism of $(V,\eta)$. The transformation (\ref{88})
is a composition of the total reflection $v\to -v$ of $V$ and a
pseudo-orthogonal reflection
\mar{ss9}\beq
v\to v-2\frac{\eta(w,v)}{\eta(w,w)}w \label{ss9}
\eeq
across a hyperplane $w^\perp$. Since $(-w)^\perp=w^\perp$, a
pseudo-orthogonal reflection across a hyperplane $w^\perp$
coincides with that across a hyperplane $(-w)^\perp$. Therefore,
the total reflection of $V$ commutes with the pseudo-orthogonal
reflection (\ref{ss9}) of $V$ across a hyperplane and, as a
consequence, with any inner automorphism $\wh w$ (\ref{88}). It
follows that any pseudo-orthogonal reflection (\ref{ss9}) of $V$
across a hyperplane is a composition of the total reflection of
$V$ and some inner automorphism (\ref{88}) of $V\subset
\cC(V,\eta)$. Since a pseudo-Euclidean space $V$ is of even
dimension, its total reflection also is an inner automorphism
\mar{ss10}\beq
 (\wh w^1\cdots \wh w^n)(v)=(w^1\cdots w^n)v(w^1\cdots w^n)^{-1}=-v, \qquad
n=\mathrm{dim}\,V. \label{ss10}
\eeq
In this case, any pseudo-orthogonal reflection (\ref{ss9}) of $V$
across a hyperplane is represented by some inner automorphism of
$V\subset \cC(V,\eta)$. By the well-known Cartan--Dieudonn\'e
theorem, every element of a pseudo-orthogonal group $O(V,\eta)$
can be written as a composition of $r\leq \mathrm{dim}\, V$
pseudo-orthogonal reflections (\ref{ss9}) across hyperplanes in
$V$ and, consequently, as a composition of inner automorphisms of
$V$. Its prolongation onto a ring $\cC(V,\eta)$ also is an inner
automorphism.
\end{remark}

Remark \ref{L2} gives something more. Let us consider a subgroup
$\mathrm{Cliff}(V,\eta)\subset \cG\cC(V,\eta)$ generated by all
invertible elements of $V\subset \cC(V,\eta)$. It is called the
Clifford group.

\begin{theorem} \mar{260} \label{260}
The homomorphism $\zeta$ (\ref{83}) of a Clifford group
$\mathrm{Cliff}(V,\eta)$ to $\mathrm{Aut}[\cC(V,\eta)]$ is its
epimorphism
\mar{103}\beq
\zeta: \cG\cC(V,\eta) \supset \mathrm{Cliff}(V,\eta)\to O(V,\eta)
\subset \mathrm{Aut}[\cC(V,\eta)]. \label{103}
\eeq
onto $O(V,\eta)$.
\end{theorem}

\begin{proof}
The transformation (\ref{88}) is an automorphism of $(V,\eta)$
and, consequently, an element of $O(V,\eta)$. Thus, the
homomorphism $\zeta$ (\ref{83}) of a Clifford group
$\mathrm{Cliff}(V,\eta)$ to $\mathrm{Aut}[\cC(V,\eta)]$ factorizes
through the homomorphism (\ref{103}). Conversely, it follows from
Remark \ref{L2} that any element of $O(V,\eta)$ is a composition
of inner automorphisms (\ref{88}) and (\ref{ss10}) which are
yielded by elements of $\mathrm{Cliff}(V,\eta)$. Consequently, the
homomorphism (\ref{103}) is an epimorphism.
\end{proof}

Due to the factorization (\ref{103}), any ring automorphism $\wh
v$, $v\in \mathrm{Cliff}(V,\eta)$, of $\cC(V,\eta)$ also is an
automorphism of a Clifford algebra $\cC(V,\eta)$. However, if
$(V',\eta')$ is a different generating space of a ring
$\cC(V,\eta)$, we have a different Clifford subgroup
$\mathrm{Cliff}(V',\eta')$ of a group $\cG\cC(V,\eta)$. Then a
Clifford group $\mathrm{Cliff}(V',\eta')$ provides ring
automorphisms of $\cC(V,\eta)$, but not automorphisms of a
Clifford algebra $\cC(V,\eta)$.

\begin{example} \label{ss6} \mar{ss6}
Let us consider a ring $\cC(2,0)= \mathrm{Mat}(2,\mathbb R)$
(\ref{6}) possessing the Euclidean basis $\{e^1,e^2\}$ (\ref{60}).
Its group of invertible elements (\ref{sp611}) is $GL(2,\mathbb
R)$. Elements of this group reads
\mar{sp620}\beq
ae +be^1 +ce^2 +de^2e^1= \begin{pmatrix}a+c & b+d \\
b-d& a- c
\end{pmatrix}. \label{sp620}
\eeq
They constitute a four-dimensional real vector space with a basis
$\{e,e^1,e^2,e^2e^1\}$. Its elements $\{e^1,e^2, e^2e^1\}$
generate a three-dimensional pseudo-Euclidean subspace $(W,\chi)$
of signature $(++-)$ such that
\be
ww'+w'w =2\chi(w,w')e, \qquad w,w'\in W.
\ee
Then any two-dimensional pseudo-Euclidean subspace $V$ of $W$ is a
generating space of a ring $\cC(2,0)$, and {\it vice versa}. The
group (\ref{83}) of automorphisms of a ring $\cC(2,0)$ is
\mar{k20}\beq
\mathrm{Aut}[\cC(2,0)]=PGL(2,\mathbb R)= SL(2,\mathbb R)/\mathbb
Z_2=SO(2,1). \label{k20}
\eeq
Any automorphism of a ring $\cC(2,0)$ is an automorphism of
$(W,\chi)$. Herewith, different automorphisms of $\cC(2,0)$ yield
the distinct ones of  $W$. Consequently, there is a monomorphism
$\mathrm{Aut}[\cC(2,0)]\to O(2,1)$. However, reflections
$e^1\to-e^1$, $e^2\to -e^2$ and $e^2e^1\to - e^2e^1$ of $W$ fail
to be ring automorphisms because they are identity automorphisms
of some two-dimensional subspaces of $W$. Therefore, we have the
monomorphism (\ref{k20}). Elements of $SO(2,1)$ are given by
compositions of automorphisms
\mar{335,a,b}\ben
&& \wh M_\alpha \begin{pmatrix} e^1 \\ e^2\\ e^2e^1
\end{pmatrix}=\begin{pmatrix}\cos\alpha &-\sin\alpha & 0\\
\sin\alpha& \cos\alpha & 0\\ 0& 0&1
\end{pmatrix} \begin{pmatrix} e^1 \\ e^2\\ e^2e^1
\end{pmatrix},\label{335} \\
&& \wh M_s \begin{pmatrix} e^1 \\ e^2\\ e^2e^1
\end{pmatrix}= \begin{pmatrix} 1 & 0 &0 \\0 & \cosh s& \sinh s\\ 0 &\sinh s
&\cosh s
\end{pmatrix}\begin{pmatrix} e^1 \\ e^2\\ e^2e^1
\end{pmatrix} \label{335a}\\
&& \wh T\begin{pmatrix} e^1 \\ e^2\\ e^2e^1
\end{pmatrix} =\begin{pmatrix} -1 & 0 &0 \\0 & 1& 0\\ 0 & 0
&-1\end{pmatrix}\begin{pmatrix} e^1 \\ e^2\\ e^2e^1
\end{pmatrix}  \label{335b}
\een
of $W$. Note that automorphisms $M_\al$ (\ref{335}) and $T$
(\ref{335b}) constitute a subgroup $O(2)\subset SO(2,1)$ of
automorphism of a Clifford algebra $\cC(2,0)$ possessing an
Euclidean generating basis $\{e^1,e^2\}$. They are inner
automorphisms (\ref{a5}) generated, e.g., by the
elements(\ref{sp620}):
\mar{sp630}
\beq
M_\al= e\cos(\al/2) -e^2e^1\sin(\al/2), \qquad T=e^2,
\label{sp630}
\eeq
of a group $O(2,\mathbb R)\subset \mathrm{Mat}(2,\mathbb R)$. It
should be however emphasized that there is no monomorphism
\be
\mathrm{Aut}[\cC(2,0)]\supset O(2,\mathbb R) \to
\mathrm{Mat}(2,\mathbb R),
\ee
whereas there exists an epimorphism
\mar{sp615}\beq
\mathrm{Mat}(2,\mathbb R)\supset O(2,\mathbb R) \to O(2,\mathbb
R)\subset \mathrm{Aut}[\cC(2,0)]. \label{sp615}
\eeq
\end{example}

\begin{example} \label{ss11} \mar{ss11}
Let us consider the Clifford algebra
$\cC(4,0)=\mathrm{Mat}(2,\mathbb H)$ (\ref{11}) whose generating
Euclidean space $V$ possesses the basis
$\{\epsilon^0,\epsilon^1,\epsilon^2,\epsilon^3\}$ (\ref{250'}).
Its elements $\{\epsilon^0,\epsilon^1,\epsilon^2,\epsilon^3,
\epsilon^5\}$ (see the notation (\ref{431})) make up a basis for a
five-dimensional Euclidean space $(W,\chi)$ such that
\be
ww'+w'w =2\chi(w,w')e, \qquad w,w'\in W.
\ee
Similarly to Example \ref{ss6}, one can show that
\be
\mathrm{Aut}[\cC(4,0)]=PGL(2,\mathbb H)=SO(5).
\ee
Due to the isomorphisms (\ref{11}), this also is the case of real
rings $\cC(1,3)$ and $\cC(0,4)$.
\end{example}

\begin{example} \label{sp633} \mar{sp633}
Let us consider the Clifford algebra
$\cC(3,1)=\mathrm{Mat}(4,\mathbb HR)$ (\ref{12}). Its automorphism
group is
\mar{sp634}\beq
\mathrm{Aut}[\cC(3,1)]=PGL(4,\mathbb R)=SL(4,\mathbb
R)/\mathbb{Z}_4. \label{sp634}
\eeq
\end{example}

\subsection{Pin and Spin groups}

The epimorphism (\ref{103}) yields an action of a Clifford group
$\mathrm{Cliff}(V,\eta)$ in a pseudo-Euclidean space $(V,\eta)$ by
the adjoint representation (\ref{a5}). However, this action is not
effective. Therefore, one consider subgroups Pin$(V,\eta)$ and
Spin$(V,\eta)$ of $\mathrm{Cliff}(V,\eta)$. The first one is
generated by elements $v\in V$ such that $\eta(v,v)=\pm 1$. A
group Spin$(V,\eta)$ is defined as an intersection
\mar{104a}\beq
\mathrm{Spin}(V,\eta)=\mathrm{Pin}(V,\eta)\cap \cC^0(V,\eta)
\label{104a}
\eeq
of a group Pin$(V,\eta)$ and the even subring $\cC^0(V,\eta)$ of a
Clifford algebra $\cC(V,\eta)$. In particular, generating elements
$v\in V$ of Pin$(V,\eta)$ do not belong to its subgroup
Spin$(V,\eta)$. Their images under the epimorphism $\zeta$
(\ref{103}) are reflections (\ref{88}) of $V$.

\begin{theorem} \mar{t1} \label{t1}
The epimorphism (\ref{103}) restricted to the Pin and Spin groups
leads to short exact sequences of groups
\mar{106,4}\ben
&& e\to \mathbb Z_2\longrightarrow
\mathrm{Pin}(V,\eta)\op\longrightarrow^\zeta O(V,\eta)\to e.
\label{106} \\
&& e\to \mathbb Z_2\longrightarrow
\mathrm{Spin}(V,\eta)\op\longrightarrow^\zeta SO(V,\eta)\to e,
\label{104}
\een
where $\mathbb Z_2\to (e,-e)\subset \mathrm{Spin}(V,\eta)$.
\end{theorem}

It should be emphasized that an epimorphism $\zeta$ in (\ref{106})
and (\ref{104}) is not a trivial bundle unless $\eta$ is of
signature $(1,1)$ (see Example \ref{e11}). It is a universal
coverings over each component of $O(V,\eta)$.

\begin{example} \mar{e10} \label{e10} Let us consider the Clifford algebra $\cC(2,0)=\mathrm{Mat}(2,\mathbb R)$ (\ref{6})
possessing the Euclidean basis $\{e^1,e^2\}$ (\ref{60}) (Example
\ref{ss6}). Then a group Pin$(2,0)$ is generated by elements
\be
a_1 e^1 + a_2 e^2, \qquad a_1^2 + a_2^2=-\det(a_1 e^1 + a_2
e^2)=1, \qquad a_1,a_2\in\mathbb R.
\ee
An even subring of $\cC(2,0)$ is represented by matrices $a e +
b\tau^2$, $a,b\in\mathbb R$. Then a group Spin$(2,0)$ consists of
elements
\be
a\mathbf{1} + b e^2e^1, \qquad \det(a e + be^2e^1)=a^2+b^2=1,
\qquad a,b\in\mathbb R,
\ee
i.e., of matrices $M_\alpha$ (\ref{sp630}) which constitute a
group $SO(2,\mathbb R)$. The epimorphism (\ref{104}) reads
\mar{282}\beq
\mathrm{Mat}(2,\mathbb R)\supset SO(2,\mathbb R) \to SO(2,\mathbb
R)\subset \mathrm{Aut}[\cC(2,0)]. \label{282}
\eeq
(cf. (\ref{sp615})). Its kernel $\zeta^{-1}(e)$ is a subgroup
$(e,-e)$ of $SO(2,\mathbb R)$.
\end{example}

\begin{example} \mar{e11} \label{e11} Let us consider the Clifford algebra
$\cC(1,1)=\mathrm{Mat}(2,\mathbb R)$ (\ref{6}). Its generating
pseudo-Euclidean space possesses a basis $\{e^1,e^1e^2\}$. Then a
group Pin$(1,1)$ is generated by elements
\be
a_1 e^1 + a_2e^1e^2, \qquad a_1^2 - a_2^2=\det(a_1e^1 +
a_2e^1e^2)=\pm 1,  \qquad a_1,a_2\in\mathbb R.
\ee
An even subring of $\cC(1,1)$ is represented by matrices $a e +
be^2$, $a,b\in\mathbb R$. Then a group Spin$(1,1)$ consists of
elements
\be
ae + be^2, \qquad \det(ae + be^2)=a^2-b^2=\pm 1, \qquad
a,b\in\mathbb R,
\ee
i.e., of elements
\be
\pm[(\cosh s)e + (\sinh s)e^2]=\pm\exp(se^2)=\pm M_s, \qquad \pm
e^2M_s.
\ee
It is isomorphic to a group $\mathbb Z_2\times\mathbb
Z_2\times\mathbb R^+$, where $R^+$ is a group of positive real
numbers. Its epimorphism $\zeta$ (\ref{104}) onto a subgroup
$SO(1,1)$ of $\mathrm{Aut}[\cC(2,0)]=SO(2,1)$ has the kernel
$(e,-e)=\mathbb Z_2\times e\times e$. It is readily observed that
$\zeta(e^2)$ is a total reflection of $\mathbb R^2$. Therefore the
exact sequence (\ref{104}) for Spin$(1,1)$ is reduced to the exact
sequence
\be
e\to \mathbb Z_2\longrightarrow
\mathrm{Spin}^+(1,1)\op\longrightarrow^\zeta SO^0(1,1)\to e
\ee
where Spin$^+(1,1)=\mathbb Z_2\times \mathbb R^+$ is a subgroup of
matrices $\pm M_s$ and $SO^0(1,1)$ is a connected component of the
unit of $SO(1,1)$.
\end{example}

\begin{example} \mar{e12} \label{e12}
Let the Clifford algebra $\cC(1,3)=\mathrm{Mat}(2,\mathbb H)$
(\ref{11}) be represented as a subalgebra of the complex Clifford
algebra $\mathbb C\cC(4)=\mathrm{Mat}(4,\mathbb C)$ (\ref{a27})
whose generating elements are Dirac's matrices
$(\wt\gamma^0,\wt\gamma^j)$ (\ref{41'}). A group Pin$(1,3)$ is
generated by  matrices
\mar{q1}\ben
&& a_\mu\wt\gamma^\mu, \quad a_0^2-\op\sum_{i=1,2,3}a_i^2=\det(a_0\mathbf{1} +a_i\sigma^i)=\pm
1,\quad a_\mu\in\mathbb R,
\label{q1}\\
&&
\det(a_\mu\wt\gamma^\mu)=\left(a_0^2-\op\sum_{i=1,2,3}a_i^2\right)^2=1.
\nonumber
\een
A group Spin$(1,3)$ is a subgroup of Pin$(1,3)$ whose elements are
even products of the matrices (\ref{q1}). It is generated by
matrices
\be
&& (a_0\wt\gamma^0 +a_i\wt\gamma^i)(b_0\wt\gamma^0
+b_i\wt\gamma^i)=(a_0\mathbf{1}
+a_i\wt\gamma^i\wt\gamma^0)(b_0\mathbf{1}
+b_i\wt\gamma^0\wt\gamma^i), \quad a_\mu,b_\mu\in\mathbb R,\\
&& \det(a_0\mathbf{1} +a_i\sigma^i)= \pm 1, \qquad \det(b_0\mathbf{1} +b_i\sigma^i)=\pm
1,\\
&& \det(a_\mu\wt\gamma^\mu)=\det(b_\mu\wt\gamma^\mu)=1.
\ee
Then elements of Spin$(1,3)$ take a form
\be
&& \begin{pmatrix} c_0 \mathbf{1}+c_i\sigma^i& 0\\0& \ol c_0
\mathbf{1}-\ol c_i\sigma^i
\end{pmatrix},  \qquad c_\mu\in\mathbb C,\\
&& \det(c_0 \mathbf{1}+c_i\sigma^i)=\det(\ol c_0
\mathbf{1}-\ol c_i\sigma^i)=\pm 1.
\ee
They read
\mar{a9}\beq
M_A=\begin{pmatrix} A & 0\\0& {\rm Tr}
A^*\mathbf{1}-A^*\end{pmatrix}, \label{a9}
\eeq
where $A$ are complex $(2\times 2)$-matrices such that
\be
\det A=\det({\rm Tr} A^*\mathbf{1}-A^*)=\pm 1.
\ee
A group Spin$(1,3)$ contains two connected components
Spin$^+(1,3)$ and Spin$^-(1,3)$ which consist of the elements
(\ref{a9}) with $\det A=1$ and $\det A=-1$, respectively. Being a
connected component of the unity, the first one is a group
$SL(2,\mathbb C)$. Elements of Spin$^-(1,3)$ come from elements of
Spin$^+(1,3)$ by means of multiplication
\be
\mathrm{Spin}^+(1,3)\ni M_A \to M_{i\mathbf 1}M_A\in
\mathrm{Spin}^-(1,3).
\ee
We have the exact sequence (\ref{104}):
\mar{105a}\beq
e\to \mathbb Z_2\longrightarrow
\mathrm{Spin}(1,3)\op\longrightarrow^\zeta SO(1,3)\to e,
\label{105a}
\eeq
where $\zeta(M_{i\mathbf 1})\in SO(1,3)$ is a total reflection.
This exact sequence is restricted to the exact sequence
\mar{105}\beq
e\to \mathbb Z_2\longrightarrow
\mathrm{Spin}^+(1,3)\op\longrightarrow^\zeta SO^0(1,3)\to e
\label{105}
\eeq
where $SO^0(1,3)$, called the proper Lorentz group, is a connected
component of the unit of $SO(1,3)$. Let us denote
\mar{300}\beq
\mathrm{L_s}=\mathrm{Spin}^+(1,3)=SL(2, \mathbb C), \qquad
\mathrm{L}=SO^0(1,3). \label{300}
\eeq
Group spaces of $\mathrm{L_s}$ and L are topological spaces
$S^3\times \mathbb R^3$ and $\mathbb R P^3\times \mathbb R^3$,
respectively.
\end{example}

\begin{example} \mar{e13} \label{e13}
Let the Clifford algebra $\cC(4,0)=\mathrm{Mat}(2,\mathbb H)$
(\ref{11}) be represented as a subalgebra of the complex Clifford
algebra $\mathbb C\cC(4)=\mathrm{Mat}(4,\mathbb C)$ (\ref{a27})
whose generating elements are the matrices
$(\wt\gamma^0,-i\wt\gamma^j)$ (\ref{41'}). A group Pin$(4,0)$ is
generated by  matrices
\mar{q1a}\ben
&& a_0\wt\gamma^0+ia_i\wt\gamma^i, \qquad \op\sum_{\mu=0,...,3}a_\mu^2=\det(a_0\mathbf{1} +a_i\tau^i)=1,
\qquad a_\mu\in\mathbb R, \label{q1a}\\
&&
\det(a_0\wt\gamma^0+ia_i\wt\gamma^i)=\left(\op\sum_{\mu=0,1,2,3}a_\mu^2\right)^2=1.
\nonumber
\een
A group Spin$(4,0)$ is a subgroup of Pin$(4,0)$ whose elements are
even products of the matrices (\ref{q1a}). It is generated by
matrices
\be
&&(a_0\wt\gamma^0 +ia_i\wt\gamma^i)(b_0\wt\gamma^0
+ib_i\wt\gamma^i)=(a_0\mathbf{1}
+ia_i\wt\gamma^i\wt\gamma^0)(b_0\mathbf{1}
+ib_i\wt\gamma^0\wt\gamma^i), \quad a_\mu,b_\mu\in\mathbb R,\\
&& \det(a_0\mathbf{1} +a_i\tau^i)= \det(b_0\mathbf{1} +b_i\tau^i)=
1,\\
&& \det(a_0\wt\gamma^0 +ia_i\wt\gamma^i)=\det(b_0\wt\gamma^0
+ib_i\wt\gamma^i)=1,
\ee
which take a form
\be
&& \begin{pmatrix} c_0 \mathbf{1}+c_i\tau^i& 0\\0& d_0
\mathbf{1}+ d_i\tau^i
\end{pmatrix},  \qquad c_\mu,d_\mu\in\mathbb R,\\
&& \det(c_0 \mathbf{1}+c_i\tau^i)=\det(d_0
\mathbf{1}+d_i\tau^i)=1.
\ee
Then elements of Spin$(4,0)$ read
\mar{a9a}\beq
\begin{pmatrix} A & 0\\0& B\end{pmatrix}, \qquad \det A=\det B=1, \label{a9a}
\eeq
where $A$, $B$ are unimodular unitary complex $(2\times
2)$-matrices. Thus, a group Spin$(4,0)$ is isomorphic to a product
$SU(2)\times SU(2)$. It contains a subgroup $(e,\gamma^0)$ such
that $\zeta(\gamma^0)$ is a total reflection of $\mathbb R^4$.
Thus, the exact sequence (\ref{104}) is reduced
\mar{301}\beq
e\to \mathbb Z_2\longrightarrow
\mathrm{Spin}^+(4,0)\op\longrightarrow^\zeta SO(3)\times SO(3)\to
e \label{301}
\eeq
where $\mathrm{Spin}^+(4,0)=\mathbb Z_2\times SU(2)/\mathbb
Z_2\times SU(2)/\mathbb Z_2$ and $\mathbb Z_2= (e,-e)$.
\end{example}

\subsection{Automorphisms of complex Clifford algebras}

Let $\mathbb C\cC(n)$ be the complex Clifford algebra (\ref{a25})
of even $n$.

\begin{theorem} \label{k21} \mar{k21}
All automorphisms of a complex Clifford algebra $\mathbb C\cC(n)$
are inner.
\end{theorem}

\begin{proof}
By virtue of Theorem \ref{k8}, there is the ring isomorphism
(\ref{k9}):
\mar{k26}\beq
\mathbb C\cC(n)=\mathrm{Mat}(2^{n/2}, \mathbb C). \label{k26}
\eeq
In accordance with Corollary \ref{k12}, this algebra is a central
simple complex algebra with the center $\cZ=\mathbb C$.  In
accordance with the above-mentioned Skolem--Noether theorem
automorphisms of these algebras are inner.
\end{proof}

\begin{theorem} \label{sp503} \mar{sp503}
Invertible elements of the Clifford algebra (\ref{k26}) constitute
a general linear group $GL(2^{n/2},\mathbb C)$.
\end{theorem}

\begin{theorem} \label{sp504} \mar{sp504}
Acting in $\mathbb C\cC(n)$ by left and right multiplications,
this group also acts in a Clifford algebra by the adjoint
representation, and we obtain an epimorphism
\be
GL(2^{n/2},\mathbb C) \to PGL(2^{n/2},\mathbb
C)=\mathrm{Aut}[\cC(n)].
\ee
and a group of its automorphisms is a projective linear group
\mar{k25}\ben
&& \mathrm{Aut}[\cC(n)]=PGL(2^{n/2}, \mathbb C)=GL(2^{n/2}, \mathbb
C)/\mathbb C=  \label{k25}\\
&& \qquad SL(2^{n/2}, \mathbb C)/\mathbb Z_{2^{n/2}}. \nonumber
\een
\end{theorem}

Automorphisms of its real subrings $\cC(n,0)$ yield automorphisms
of $\mathbb C\cC(n)$, but do not exhaust all automorphisms of
$\mathbb C\cC(n)$ (Theorem \ref{ss26}).

Let us note that automorphisms under discussions need not be
automorphisms of $\mathbb C\cC(n)$ as an involutive algebra
(Remark \ref{ss30})

Any automorphism $g$ of a complex Clifford algebra $\mathbb
C\cC(n)$ sends its Euclidean generating space $(\cV,\kappa)$ onto
some generating space
\be
(\cV',\kappa'), \qquad \kappa'(g(v),g(v'))=\kappa(v,v'), \qquad
v,v'\in \cV,
\ee
which is the Euclidean one with respect to the basis $\{g(e^i\}$.
If $g$ preserves $\cV$, then
\be
\kappa(g(v),g(v'))=\kappa(v,v'), \qquad v,v'\in \cV,
\ee
i.e., $g$ is an automorphism of a metric space $(\cV,\kappa)$.

Conversely, any automorphism
\be
 g:\cV\ni v\to gv\in \cV, \qquad
\kappa(g(v),g(v'))=\kappa(v,v'), \qquad g\in O(n,\mathbb C),
\ee
of an Euclidean generating space $(\cV,\kappa)$ is prolonged to an
automorphism of a ring $\mathbb C\cC(n)$. Then we have a
monomorphism
\mar{ss25}\beq
O(n,\mathbb C)\to \mathrm{Aut}[\mathbb C\cC(n)] \label{ss25}
\eeq
of a group $O(n,\mathbb C)$ of automorphisms of an Euclidean
generating space $(\cV,\kappa)$ to a group of ring automorphisms
of $\mathbb C\cC(n)$. Herewith, an automorphism $g\in O(n,\mathbb
C)$ of a complex ring $\mathbb C\cC(n)$ is the identity one iff
its restriction to $\cV$ is an identity map of $\cV$.
Consequently, the following is true.

\begin{theorem} \label{ss26} \mar{ss26}
All ring automorphisms of a complex Clifford algebra $\mathbb
C\cC(n)$ preserving its Euclidean generating space constitute a
group $O(n,\mathbb C)$.
\end{theorem}

Let $Z\mathbb C\cC(n)$ denote a set of Euclidean generating spaces
of a complex Clifford algebra $\mathbb C\cC(n)$. If $n>1$, a set
$Z\mathbb C\cC(n)$ contains more than one element. Indeed, let
$\cV$ be a generating space of $\mathbb C\cC(n)$ spanned by its
Euclidean basis $\{e^1,\ldots, e^n\}$. Then, $\{e^1, ie^1e^2,
\ldots, ie^1e^n\}$ is an Euclidean basis for a different
generating space of $\mathbb C\cC(n)$.

\begin{lemma} \mar{468} \label{468} A group Aut$[\mathbb C\cC(n)]$ of ring automorphism of a
complex Clifford algebra $\mathbb C\cC(n)$ acts in a set $Z\mathbb
C\cC(n)$ effectively and transitively, i.e., no element
Aut$[\mathbb C\cC(n)]\ni g\neq e$ is the identity morphism of
$Z\mathbb C\cC(n)$ and, for any two different elements of
$Z\mathbb C\cC(n)$, there exists a ring automorphism of $\mathbb
C\cC(n)$ which sends them onto each other.
\end{lemma}

\begin{proof}
Let an automorphism Aut$[\mathbb C\cC(n)]\ni g\neq e$ preserves
some Euclidean generating space $(\cV,\kappa)$ of $\mathbb
C\cC(n)$. There exists an element $v\in \cV$ such that $v^2=e$ and
$g(v)\neq v$. Then $\cV$ admits an Euclidean basis
$\{v,e^2,\ldots, e^n\}$, and there exists a different generating
space of $\mathbb C\cC(n)$ possessing an Euclidean basis
$\{v,ve^2,\ldots, ve^n\}$. It is not preserved by an automorphism
$g$. Let $\cV$ and $\cV'$ be two different generating spaces of
$\mathbb C\cC(n)$ with Euclidean bases $\{e^1,\ldots,e^n\}$ and
$\{e'^1,\ldots,e'^n\}$, respectively. Then an association $e^i\to
e'^i$ provides an isomorphism $\cV\to \cV'$ which is extended to
an automorphism of $\mathbb C\cC(n)$.
\end{proof}

It follows from Theorem \ref{ss26} and Lemma \ref{468} that, if
$n>1$, a set $Z\mathbb C\cC(n)$ of generating spaces of a complex
Clifford algebra $\mathbb C\cC(n)$ is a homogeneous space
\mar{469}\beq
Z\mathbb C\cC(n)= PGL(2^{n/2}, \mathbb C)/O(n,\mathbb C).
\label{469}
\eeq

Given a complex Clifford algebra $\mathbb C\cC(n)$, let
$\cC(m,n-m)$ be a real Clifford algebra. Due to the canonical ring
monomorphism $\cC(m,n-m)\to \mathbb C\cC(n)$ (\ref{sp200}), there
is the canonical group monomorphism
\mar{sp201}\beq
\cG\cC(m,n-m)\to \mathrm{Mat}(2^{n/2},\mathbb C). \label{sp201}
\eeq
Since all ring automorphisms of a Clifford algebra are inner
(Theorem \ref{k15}), they are extended to inner automorphisms of a
complex Clifford algebra $\mathbb C\cC(n)$ and, consequently,
there is a group monomorphism
\mar{sp202}\beq
\mathrm{Aut}[\cC(m.n-m)]\to PGL(2^{n/2},\mathbb C). \label{sp202}
\eeq

\begin{example} \label{k27} \mar{k27}
Let us consider the complex Clifford algebra $\mathbb
C\cC(2)=\mathrm{Mat}(2, \mathbb C)$ (\ref{a26}). It possesses an
Euclidean basis $\{e^1,e^2\}$ obeying the relations (\ref{210}).
Its elements $\{e^1,e^2, ie^1e^2\}$ form a basis for a
three-dimensional complex subspace $W$ of $\mathbb C\cC(2)$
provided with a non-degenerate bilinear form $\chi$ such that
\be
ww'+w'w=2\chi(w,w')e, \qquad w,w'\in W.
\ee
Then any two-dimensional generating space of $\mathbb C\cC(2)$ is
a subspace of $W$, and any ring automorphism of $\mathbb C\cC(2)$
is that of $W$. By virtue of Theorem \ref{k21}, the group of
automorphisms of $\mathbb C\cC(2)$ is
\mar{452}\beq
\mathrm{Aut}[\mathbb C\cC(2)]\to SL(2,\mathbb C)/\mathbb
Z_2=SO(3,\mathbb C). \label{452}
\eeq
Elements of $SO(3,\mathbb C)$ are given by compositions of
automorphisms
\mar{352,'}\ben
&& \wh a_{\phi s}\begin{pmatrix}e^1\\ e^2\\ e^1e^2
\end{pmatrix}=\begin{pmatrix}\cos(\phi+is) &-\sin(\phi+is) & 0\\
\sin(\phi+is)& \cos(\phi+is) & 0\\ 0& 0&1
\end{pmatrix}\begin{pmatrix}e^1\\ e^2\\ e^1e^2
\end{pmatrix}, \label{352}\\
&&  \wh a_{\theta r}\begin{pmatrix}e^1\\ e^2\\ e^1e^2
\end{pmatrix}=\begin{pmatrix} 1 &0 & 0 \\ 0& \cos(\theta +ir) & i\sin(\theta
+ir) \\ 0 & i\sin(\theta +ir) & \cos(\theta +ir)
\end{pmatrix}\begin{pmatrix}e^1\\ e^2\\ e^1e^2
\end{pmatrix} \label{352'}
\een
of $W$. In accordance with the relation (\ref{469}) and the
isomorphism (\ref{452}), we obtain a set
\mar{470}\beq
Z\mathbb C\cC(2)= SO(3, \mathbb C)/O(2,\mathbb C) \label{470}
\eeq
of generating spaces of a complex Clifford algebra $\mathbb
C\cC(2)$. The automorphisms (\ref{352}) -- (\ref{352'}) are inner
automorphisms
\mar{355,6,7}\ben
&&\wh a_{\phi s}(a)= a_{\phi s}a a_{\phi s}^{-1}, \qquad \wh
a_{\theta r}(a)= a_{\theta r}a a_{\theta r}^{-1}, \qquad a\in
\mathbb C\cC(2), \label{355}\\
&& a_{\phi s}= e\cos(\phi/2+is/2) +
e^1e^2\sin(\phi/2+is/2), \label{356} \\
&& a_{\phi s}^{-1}=
e\cos(\phi/2+is/2)- e^1e^2\sin(\phi/2+is/2), \nonumber\\
&& a_{\theta r}= e\cos(\theta/2+ir/2) +
ie^1\sin(\theta/2+ir/2),\label{357}
\\
&& a_{\theta r}^{-1}= e\cos(\theta/2+ir/2) -
ie^1\sin(\theta/2+ir/2), \nonumber
\een

In particular, the Spin groups Spin$(2,0)$, Spin$(0,2)$ and
Spin$(1,1)$ in Examples \ref{e10} and \ref{e11} yield inner
automorphism $\wh a_{\phi,0}$ and $\wh a_{0,s}$ (\ref{352}) of
$\mathbb C\cC(2)$, respectively. There are natural injections of a
group $SO(2)=\zeta(\mathrm{Spin}(2,0))$ of automorphisms of a
Clifford algebra $\cC(2,0)$, a group
$SO(1,1)=\zeta(\mathrm{Spin}(1,1))$ of automorphisms of a Clifford
algebra $\cC(1,1)$ and a group $SO(2)=\zeta(\mathrm{Spin}(0,2))$
of automorphisms of a Clifford algebra $\cC(0,2)$ to $SO(3,\mathbb
C)$.
\end{example}

\begin{remark} \label{ss30} \mar{ss30}
Let us note that automorphisms $\wh a_{\phi s\neq 0}$ (\ref{352})
and $\wh a_{\theta r\neq 0}$ (\ref{352'}) do not transform
Hermitian elements to Hermitian elements, and thus they are not
automorphisms of an involutive algebra $\mathbb C\cC(2)$.
\end{remark}

\begin{example} \label{k30} \mar{k30}
Let us consider the complex Clifford algebra $\mathbb
C\cC(4)=\mathrm{Mat}(4, \mathbb C)$ (\ref{a27}) possessing an
Euclidean basis $\{\epsilon^\mu\}$. Its elements
$\{\epsilon^\mu,\epsilon^5\}$ form a basis for a five-dimensional
complex subspace $W$ of $\mathbb C\cC(4)$ provided with a
non-degenerate bilinear form $\chi$ such that
\be
ww'+w'w=2\chi(w,w')e, \qquad w,w'\in W.
\ee
Then any four-dimensional complex subspace $\cV$ of $W$ provided
with an induced bilinear form is a generating space of a complex
Clifford algebra $\mathbb C\cC(4)$. By virtue of Theorem
\ref{k21}, the group of automorphisms of $\mathbb C\cC(4)$ is
\mar{460}\beq
\mathrm{Aut}[\mathbb C\cC(4)]=PGL(4,\mathbb C)=SO(6,\mathbb
C)/\mathbb Z_2. \label{460}
\eeq
Then in accordance with the relation (\ref{469}), we obtain a set
\mar{471}\beq
Z\mathbb C\cC(4)= PGL(4,\mathbb C)/O(4,\mathbb C) \label{471}
\eeq
of generating spaces of a complex Clifford algebra $\mathbb
C\cC(4)$.
\end{example}

\section{Spinor spaces of complex Clifford algebras}

As was mentioned above, we define spinor spaces in terms of
Clifford algebras.

\begin{definition} \label{sp510} \mar{sp510}
A real spinor space $\Psi(m,n-m)$ is defined as a carrier space of
an irreducible representation of a Clifford algebra $\cC(m,n-m)$.
\end{definition}

It also carries out a representation of the corresponding group
Spin$(m,n-m)\subset \cC(m,n-m)$ \cite{law}.

If $n$ is even, such a real spinor space is unique up to an
equivalence in accordance with Theorem \ref{a10}. However,
Examples \ref{r88} -- \ref{r8} of Clifford algebras $\cC(0,2)$ and
$\cC(2,0)$, respectively, show that spinor spaces $\Psi(m,n-m)$
and $\Psi(m',n-m')$ need not be isomorphic vector spaces for
$m'\neq m$.

For instance, a Dirac spinor space is defined to be a spinor space
$\Psi(1,3)$ of a Clifford algebra $\cC(1,3)$ (Example \ref{k1}).
It differs from a Majorana spinor space $\Psi(3,1)$ of a Clifford
algebra $\cC(3,1)$ (Example \ref{sp600}). In contrast with the
four-dimensional real matrix representation (\ref{62}) of
$\cC(3,1)$, a representation of a real Clifford algebra $\cC(3,1)$
by complex Dirac's matrices (\ref{41a}) is not a representation a
real Clifford algebra by virtue of Definition \ref{sp602}. From
the physical viewpoint, Dirac spinor fields describing charged
fermions are complex fields.

Therefore, we consider complex spinor spaces.

\begin{definition} \label{ss34} \mar{ss34}
A complex spinor space $\Psi(n)$ is defined as a carrier space of
an irreducible representation of a complex Clifford algebra
$\mathbb C\cC(n)$.
\end{definition}

Since $n$ is even, a representation $\Psi(n)$ is unique up to an
equivalence in accordance with Theorem \ref{a11}. Therefore, it is
sufficient to describe a complex spinor space $\Psi(n)$ as a
subspace of a complex Clifford algebra $\mathbb C\cC(n)$ which
acts on $\Psi(n)$ by left multiplications.

Given a complex Clifford algebra $\mathbb C\cC(n)$, let us
consider its non-zero minimal left ideal which $\cC(n)$ acts on by
left multiplications. It is a finite-dimensional complex vector
space. Therefore, an action of a complex Clifford algebra $\mathbb
C\cC(n)$ in a minimal left ideal by left multiplications defines a
linear representation of $\mathbb C\cC(n)$. It obviously is
irreducible. In this case, a minimal left ideal of $\mathbb
C\cC(n)$ is a complex spinor space $\Psi(n)$. Thus, we come to an
equivalent definition of a complex spinor space.

\begin{definition} \label{ss33} \mar{ss33}
Complex spinor spaces $\Psi(n)$ are minimal left ideals of a
complex Clifford algebra $\mathbb C\cC(n)$ which carry out its
irreducible representation (\ref{sp110}).
\end{definition}

By virtue of Theorem \ref{k8}, there is a ring isomorphism
$\mathbb C\cC(n)=\mathrm{Mat}(2^{n/2}, \mathbb C)$ (\ref{k26}).
Then we come to the following.

\begin{theorem} \label{k53} \mar{k53}
A spinor representation of a Clifford algebra $\mathbb C\cC(n)$ is
equivalent to the canonical representation of
$\mathrm{Mat}(2^{n/2}, \mathbb C)$ by matrices in a complex vector
space $\mathbb C^{2^{n/2}}$, i.e., $\Psi(n)=\mathbb C^{2^{n/2}}$.
\end{theorem}

\begin{corollary} \label{k54} \mar{k54}
A spinor space $\Psi(n)\subset \mathbb C\cC(n)$ also carries out
the left-regular irreducible representation of a group
$GL(2^{n/2},\mathbb C)=\cG\mathbb C\cC(n)$ which is equivalent to
the natural matrix representation of $GL(2^{n/2},\mathbb C)$ in
$\mathbb C^{2^{n/2}}$.
\end{corollary}

\begin{corollary} \label{sp203} \mar{sp203}
Owing to the monomorphism $\cC(m,n-m)\to \mathbb C\cC(n)$
(\ref{sp200}), a spinor space $\Psi(n)$ also carries out a
representation of real Clifford algebras $\cC(m,n-m)$, their
Clifford< Pin and Spin groups, though these representation need
not be reducible.
\end{corollary}

In order to describe complex spinor spaces in accordance with
Definition \ref{ss33}, we are based on the following (Theorem
\ref{a41}).

\begin{lemma} \mar{a20} \label{a20} If a minimal left ideal $Q$ of a
complex Clifford algebra $\mathbb C\cC(n)$ is generated by an
element $q$, then $q^2=\lambda q$, $\lambda\in \mathbb C$.
\end{lemma}

\begin{proof} Let $q\in Q$ such that
$q^2\neq\lambda q$, $\lambda\in \mathbb C$. There are two
variants: (i) $q^N=0$ starting with some natural number $N>2$,
(ii) there is no $m>2$ such that $q^m=0$. In the first case, let
us consider a left ideal $Q'$ generated by $q^{N-1}=q^{N-2}q\in
Q$. It does not contains $q$ because, if $q=bq^{N-1}$, $b\in
\mathbb C\cC(n)$, then $q^2=bq^N=0$ that contradicts the condition
$q^2\neq\lambda q$. Thus, a left ideal $Q'$ is a proper subset of
$Q$, i.e., $Q$ fails to be minimal. In the second variant, since
$q^2\neq\lambda q$ and $Q$ is a finite-dimensional complex space,
there exists a natural number $m>2$ such that elements $q^r$,
$r=2,\ldots, m,$ are linearly dependent, i.e.,
\be
\op\sum_{r=2}^m\lambda_r q^r=0, \qquad \lambda_r\in\mathbb C.
\ee
This equality is brought into the form $q^pc=cq^p=0$, $c\in
\mathbb C\cC(n)$ for some $1< p <m$. Let us consider a left ideal
$Q'$ generated by an element $cq^{p-1}=(cq^{p-2})q\in Q$. It does
not contains an element $q$ because, if $q=bcq^{p-1}$, $b\in
\mathbb C\cC(n)$, then $q^2=bcq^p=0$ that contradicts the
condition $q^2\neq\lambda q$. Thus, a left ideal $Q'$ is a proper
subset of $Q$, i.e., $Q$ fails to be minimal.
\end{proof}

Lemma \ref{a20} gives something more. Since a minimal left ideal
$Q$ of $\mathbb C\cC(n)$ is generated by any its element, each
element $q\in Q$ possesses a property $q^2=\lambda q$, $\lambda\in
\mathbb C$.

\begin{lemma} \mar{a30} \label{a30}
A minimal left ideal of a complex Clifford algebra $\mathbb
C\cC(n)$ contains a non-zero Hermitian element, and thus it is
generated by a Hermitian element.
\end{lemma}

\begin{proof}
Let $q\neq 0$ be an element of $Q$. Then, $q^*q\neq 0$ in
accordance with the inequality (\ref{a34}), and this is a
Hermitian element of $Q$.
\end{proof}

By virtue of Lemmas \ref{a20} -- \ref{a30}, any minimal left ideal
of a complex Clifford algebra $\mathbb C\cC(n)$ is generated by a
Hermitian idempotent $p=p^*$, $p^2=p$. Of course, it is not
invertible because invertible elements generate an algebra
$\mathbb C\cC(n)$ which contains proper left ideals. It is readily
observed that any Hermitian idempotent takes a form
\mar{a21}\beq
p=\frac12(e + s), \qquad s^2=e, \qquad s^*=s, \qquad s\neq e.
\label{a21}
\eeq

Thus, the following has been proved.

\begin{theorem} \mar{a41} \label{a41}
Any complex spinor space $\Psi(n)$ is generated by some Hermitian
idempotent $p\in \Psi(n)$ (\ref{a21}).
\end{theorem}

The converse however need not be true.

\begin{example} \label{k61} \mar{k61}
Let us consider a Hermitian idempotent $p\in \mathrm{Mat}(2^{n/2},
\mathbb C)$ whose non-zero component is only $p_{11}=1$. It
generates a minimal left ideal $\Psi_{11}(n)\subset
\mathrm{Mat}(2^{n/2}, \mathbb C)$ which consists of matrices $a\in
\mathrm{Mat}(2^{n/2}, \mathbb C)$ whose columns, except $a_{1i}$
equal zero.
\end{example}

Certainly, an automorphism of a Clifford algebra $\mathbb C\cC(n)$
sends a spinor space onto a spinor space.

\begin{lemma} \label{k60} \mar{k60} An action of a group $PGL(2^{n/2},\mathbb C)$
of automorphisms of $\mathbb C\cC(n)$ in a set $S\Psi(n)$ of
spinor spaces is transitive.
\end{lemma}

\begin{proof} Let $\Psi(n)$ and $\Psi'(n)$ are spinor space
defined by Hermitian idempotents $p\in \mathrm{Mat}(2^{n/2},
\mathbb C)$ and $p\in \mathrm{Mat}(2^{n/2}, \mathbb C)$. Since
right-regular representation of a group $GL(2^{n/2},\mathbb C)$ in
$\mathrm{Mat}(2^{n/2}, \mathbb C)$ is transitive, there exists an
element $g\in \mathrm{Mat}(2^{n/2}, \mathbb C)$ so that
$p'=pg^{-1}$. Then a spinor space $\Psi'(n)$ is generated by an
idempotent $gpg^{-1}$, and an inner automorphism $a\to gpg^{-1}$,
$a\in \mathbb C\cC(n)$, sends $\Psi(n)$ onto $\Psi'(n)$.
\end{proof}

Given a spinor space $\Psi(n)$, let $G\Psi(n)$ be a subgroup of
$PGL(2^{n/2},\mathbb C)$ which preserves $\Psi(n)$. Then it
follows from Lemma \ref{k60} that a set of spinor spaces
$S\Psi(n)$ is bijective to the quotient
\mar{k62}\beq
S\Psi(n)= PGL(2^{n/2},\mathbb C)/G\Psi(n). \label{k62}
\eeq

For instance, let $\Psi_{11}(n)$ be a spinor space in Example
\ref{k61}. Its stabilizer $G\Psi_{11}(n)$ consists of inner
automorphisms generated by elements $g\in \mathrm{Mat}(2^{n/2},
\mathbb C)$ with components $g_{k1}=0$, $1<k$.

\section{Reduced structures}

This section addresses gauge theory on principal bundles in a case
of spontaneous symmetry breaking \cite{book09,book00,book13}.

\subsection{Reduced structures in gauge theory}

 Let $G$ be a real Lie group
whose unit is denoted by $\bb$. A fibre bundle
\mar{51f1}\beq
\pi_P :P\to X \label{51f1}
\eeq
is called a principal bundle with a structure group $G$ if it
admits an action of $G$ on $P$ on the right by a fibrewise
morphism
\mar{1}\beq
R_P: G\op\times_X P \ar_X P, \qquad R_g: p\to pg, \qquad
\pi_P(p)=\pi_P(pg), \quad p\in P,  \label{1}
\eeq
which is free and transitive on each fibre of $P$. It follows
that:

$\bullet$ a typical fibre of $P$ (\ref{51f1}) is a group space of
$G$, which a structure group $G$ acts on by left multiplications;

$\bullet$ the quotient of $P$ with respect to the action (\ref{1})
of $G$ is diffeomorphic to a base $X$, i.e., $P/G=X$;

$\bullet$ a principal bundle $P$ is equipped with a bundle atlas
\mar{51f2}\beq
\Psi_P=\{(U_\al,\psi^P_\al),\vr_{\al\bt}\} \label{51f2}
\eeq
whose trivialization morphisms
\be
\psi_\al^P: \pi_P^{-1}(U_\al)\to U_\al\times G
\ee
obey a condition
\mar{tt1}\beq
\psi_\al^P(pg)=g\psi_\al^P(p), \qquad g\in G, \label{tt1}
\eeq
and transition functions $\vr_{\al\bt}$ are local $G$-valued
functions.

For short, we call $P$ (\ref{51f1})  the principal $G$-bundle.

Due to the property (\ref{tt1}), every trivialization morphism
$\psi^P_\al$ determines a unique local section $z_\al:U_\al\to P$
such that
\be
(\psi^P_\al\circ z_\al)(x)=\bb, \qquad x\in U_\al.
\ee
A transformation law for $z_\al$ reads
\mar{b1.202}\beq
z_\bt(x)=z_\al(x)\vr_{\al\bt}(x),\qquad x\in U_\al\cap
U_\bt.\label{b1.202}
\eeq
Conversely, a family
\mar{vcv}\beq
\Psi_P=\{(U_\al,z_\al),\vr_{\al\bt}\} \label{vcv}
\eeq
of local sections of $P$ which obey the transformation law
(\ref{b1.202}) determines the unique bundle atlas $\Psi_P$
(\ref{51f2}) of a principal bundle $P$.

\begin{corollary} \label{52a3} \mar{52a3}
It follows that a principal bundle admits a global section iff it
is trivial.
\end{corollary}

\begin{example} \mar{52e1} \label{52e1} Let $H$ be a closed
subgroup of a real Lie group $G$. Then $H$ is a Lie group. Let
$G/H$ be the quotient of $G$ with respect to an action of $H$ on
$G$ by right multiplications. Then
\mar{ggh}\beq
\pi_{GH}:G\to G/H \label{ggh}
\eeq
 is a principal $H$-bundle. If $H$ is a maximal
compact subgroup of $G$, the quotient $G/H$ is diffeomorphic to an
Euclidean manifold $\mathbb R^m$ and the principal bundle
(\ref{ggh}) is trivial, i.e., $G$ is diffeomorphic to a product
$\mathbb R^m\times H$.
\end{example}

\begin{remark} \mar{52r1} \label{52r1}
If $f:X'\to X$ is a manifold morphism,  the pull-back $f^*P\to X'$
 of a principal bundle also is a principal bundle with the same
structure group as of $P$.
\end{remark}

\begin{remark}  \mar{52r1'} \label{52r1'}
Let $P\to X$ and $P'\to X'$ be principal $G$- and $G'$-bundles,
respectively. A bundle morphism $\Phi:P\to P'$ is a morphism of
principal bundles if there exists a Lie group homomorphism
$\g:G\to G'$ such that $\Phi(pg)=\Phi(p)\g(g)$.
\end{remark}

 In accordance with Remark
\ref{52r1'}, an automorphism $\Phi_P$ of a principal $G$-bundle
$P$ is called principal if it is equivariant under the right
action (\ref{1}) of a structure group $G$ on $P$, i.e.,
\mar{55ff1}\beq
\Phi_P(pg)=\Phi_P(p)g, \qquad g\in G, \qquad p\in P. \label{55ff1}
\eeq
In particular, every vertical principal automorphism of a
principal bundle $P$ is represented as
\mar{b3111}\beq
\Phi_P(p)=pf(p), \qquad p\in P, \label{b3111}
\eeq
where $f$ is a $G$-valued equivariant function on $P$, i.e.,
\mar{b3115}\beq
f(pg)=g^{-1}f(p)g, \qquad g\in G. \label{b3115}
\eeq
Note that there is one-to-one correspondence
\mar{56f2}\beq
s(\pi_P(p))p = pf(p), \qquad p\in P, \label{56f2}
\eeq
between the equivariant functions $f$ (\ref{b3115}) (consequently,
the vertical automorphisms of $P$) and the global sections $s$ of
the group bundle $P^G$ (\ref{b3130}) (Example \ref{56e2}).

Let $P$ (\ref{51f1}) be a principal bundle and $V$ a smooth
manifold that on a group $G$ acts on the left. Let us consider the
quotient
\mar{b1.230}\beq
Y=(P\times V)/G \label{b1.230}
\eeq
of a product $P\times V$ by identification of elements $(p,v)$ and
$(pg,g^{-1}v)$ for all $g\in G$. It is a fibre bundle with a
structure group $G$ and a typical fibre $V$ which is said to be
associated to be associated to the principal $G$-bundle $P$. For
the sake of brevity, we call it the $P$-associated bundle.

\begin{example} \label{56e2} \mar{56e2}
A $P$-associated group bundle is defined as the quotient
\mar{b3130}\beq
\pi_{P^G}:P^G =(P\times G)/G\to X, \label{b3130}
\eeq
where a structure group $G$ which acts on itself by the adjoint
representation. There is the following fibre-to-fibre action of
the group bundle $P^G$ on any $P$-associated bundle $Y$
(\ref{b1.230}):
\be
P^G\op\times_X Y\op\to_X Y, \qquad ((p, g)/G, (p, v)/G) \to (p,
gv)/ G, \qquad g\in G, \quad  v\in V.
\ee
For instance, the action of $P^G$ on $P$ in the formula
(\ref{56f2}) is of this type.
\end{example}

The peculiarity of the $P$-associated bundle $Y$ (\ref{b1.230}) is
the following.

$\bullet$ Every bundle atlas $\Psi_P=\{(U_\al, z_\al)\}$
(\ref{vcv}) of $P$ defines a unique associated bundle atlas
\mar{aaq1}\beq
\Psi=\{(U_\al,\psi_\al(x)=[z_\al(x)]^{-1})\} \label{aaq1}
\eeq
of the quotient $Y$ (\ref{b1.230}).

$\bullet$ Any principal automorphism $\Phi_P$ (\ref{55ff1}) of $P$
yields a unique principal automorphism
\mar{024}\beq
\Phi_Y: (p,v)/G\to  (\Phi_P(p),v)/G, \qquad p\in P, \qquad v\in V,
\label{024}
\eeq
of the $P$-associated bundle $Y$ (\ref{b1.230}).

\begin{remark} \label{ss40} \mar{ss40}
In classical gauge theory on a principal bundle $P$, matter fields
are described as sections of $P$-associated bundles
(\ref{b1.230}).
\end{remark}

As was mentioned above, spontaneous symmetry breaking in classical
gauge theory on a principal bundle $P\to X$ is characterized by a
reduction of a structure group of $P$ \cite{book09,higgs,tmp}

Let $H$ and $G$ be Lie groups and $\f:H\to G$ a Lie group
homomorphism. If $P_H\to X$ is a principal $H$-bundle, there
always exists a principal $G$-bundle $P_G\to X$ together with the
principal bundle morphism
\mar{510f1}\beq
\Phi:P_H\ar_X P_G \label{510f1}
\eeq
over $X$ (Remark \ref{52r1'}). It is a $P_H$-associated bundle
\be
P_G=(P_H\times G)/H
\ee
with a typical fibre $G$ which on $H$ acts on the left by the rule
$h(g)=\f(h)g$, while $G$ acts on $P_G$ as
\be
G\ni g': (p,g)/ H \to (p,gg')/ H.
\ee

Conversely, if $P_G\to X$ is a principal $G$-bundle, a problem is
to find a principal $H$-bundle $P_H\to X$ together with the
principal bundle morphism (\ref{510f1}). If $H\to G$ is a group
epimorphism, one says that $P_G$ gives rise to $P_H$. If $H\to G$
is a closed subgroup, we have the  structure group reduction. In
this case, the bundle monomorphism (\ref{510f1}) is called a
reduced $H$-structure.

Let $P$ (\ref{51f1}) be a principal $G$-bundle,  and let $H$,
dim$H>0$, be a closed subgroup of $G$. Then we have a composite
bundle
\mar{b3223a}\beq
P\to P/H\to X, \label{b3223a}
\eeq
where
\mar{b3194}\beq
 P_\Si=P\ar^{\pi_{P\Si}} P/H \label{b3194}
\eeq
is a principal bundle with a structure group $H$ and
\mar{b3193}\beq
\Si=P/H\ar^{\pi_{\Si X}} X \label{b3193}
\eeq
is a $P$-associated bundle with a typical fibre $G/H$ which on a
structure group $G$ acts on the left (Example \ref{52e1}).

\begin{definition} \label{ss41} \mar{ss41}
One says that a structure Lie group $G$ of a principal bundle
$P$ is reduced to its closed subgroup $H$ if the following
equivalent conditions hold.

$\bullet$ A principal bundle $P$ admits a bundle atlas $\Psi_P$
(\ref{51f2}) with $H$-valued transition functions $\vr_{\al\bt}$.

$\bullet$ There exists a reduced principal subbundle $P_H$ of $P$
with a structure group $H$.
\end{definition}

\begin{remark} \mar{rttf} \label{rttf}
It is easily justified that these conditions are equivalent. If
$P_H\subset P$ is a reduced principal subbundle, its atlas
(\ref{vcv}) given by local sections $z_\al$ of $P_H\to X$ is a
desired atlas of $P$. Conversely, let $\Psi_P=\{(U_\al,
z_\al),\vr_{\al\bt}\}$ (\ref{vcv}) be an atlas of $P$ with
$H$-valued transition functions $\vr_{\al\bt}$. For any $x\in
U_\al\subset X$, let us define a submanifold $z_\al(x)H\subset
P_x$. These submanifolds form a desired $H$-subbundle of $P$
because
\be
z_\al(x)H=z_\bt(x)H\vr_{\bt\al}(x)
\ee
on the overlaps $U_\al\cap U_\bt$.
\end{remark}

A key point is the following.

\begin{theorem}\label{redsub} \mar{redsub}
There is one-to-one correspondence
\mar{510f2}\beq
P^h=\pi_{P\Si}^{-1}(h(X)) \label{510f2}
\eeq
between the reduced principal $H$-subbundles $i_h:P^h\to P$ of $P$
and the global sections $h$ of the quotient bundle $P/H\to X$
(\ref{b3193}) \cite{book09,kob}.
\end{theorem}

In classical field theory, global sections of a quotient bundle
$P/H\to X$ are interpreted as  classical Higgs fields
\cite{book09,sard92,higgs,tmp}.

\begin{corollary} \mar{5101c} \label{5101c}
A glance at the formula (\ref{510f2}) shows that a reduced
principal $H$-bundle $P^h$ is the restriction $h^*P_\Si$ of a
principal $H$-bundle $P_\Si$ (\ref{b3194}) to $h(X)\subset \Si$.
Any atlas $\Psi^h$ of a principal $H$-bundle $P^h$ defined by a
family of local sections of $P^h\to X$ also is an atlas of a
principal $G$-bundle $P$ and the $P$-associated bundle $\Si\to X$
(\ref{b3193}) with $H$-valued transition functions (Remark
\ref{rttf}). Herewith, a Higgs field $h$ written with respect to
an atlas $\Psi^h$ takes its values into the center of a quotient
$G/H$.
\end{corollary}

\begin{remark} \label{sp220} \mar{sp220}
Let $P^h$ be a reduced principal $H$-subbundle of a principal
$G$-bundle in Corollary \ref{5101c}. Any principal automorphism
$g\f$ of $P^h$ gives rise to a principal automorphism of $P$ by
means of the relation $\f(P^hg)=\f(P^h)g$, $g\in G$.
\end{remark}

In general, there is topological obstruction to reduction of a
structure group of a principal bundle to its subgroup.

\begin{theorem} \label{510a1} \mar{510a1}
The structure group $G$ of a principal bundle $P$ always is
reducible to its closed subgroup $H$, if the quotient $G/H$ is
diffeomorphic to a Euclidean space $\mathbb R^m$.
\end{theorem}

In particular, this is the case of a maximal compact subgroup $H$
of a Lie group $G$ (Example \ref{52e1}). Then the following is a
corollary of Theorem \ref{510a1} \cite{ste}.

\begin{theorem} \label{comp} \mar{comp}
A structure group $G$ of a principal bundle always is reducible to
its maximal compact subgroup $H$.
\end{theorem}

Given different Higgs fields $h$ and $h'$, the corresponding
principal $H$-subbundles $P^h$ and $P^{h'}$ of a principal
$G$-bundle $P$ fail to be isomorphic to each other in general
\cite{book09,higgs}.

\begin{theorem}\label{iso1} \mar{iso1} Let a structure Lie group $G$
of a principal bundle be reducible to its closed subgroup $H$.

$\bullet$ Every vertical principal automorphism $\Phi$ of $P$
sends a reduced principal $H$-subbundle $P^h$ of $P$ onto an
isomorphic principal $H$-subbundle $P^{h'}$.

$\bullet$ Conversely, let two reduced principal subbundles $P^h$
and $P^{h'}$ of a principal bundle $P\to X$ be isomorphic to each
other, and let $\Phi:P^h\to P^{h'}$ be their isomorphism over $X$.
Then $\Phi$ is extended to a vertical principal automorphism of
$P$.
\end{theorem}

\begin{proof}
Let
\mar{510ff}\beq
\Psi^h=\{(U_\al,z^h_\al), \vr^h_{\al\bt}\} \label{510ff}
\eeq
be an atlas of a reduced principal subbundle $P^h$, where
$z^h_\al$ are local sections of $P^h\to X$  and $\vr^h_{\al\bt}$
are the transition functions. Given a vertical automorphism $\Phi$
of $P$, let us provide a subbundle $P^{h'}=\Phi(P^h)$ with an
atlas
\mar{510ff1}\beq
\Psi^{h'}=\{(U_\al,z^{h'}_\al), \vr^{h'}_{\al\bt}\} \label{510ff1}
\eeq
given by the local sections $z^{h'}_\al =\Phi\circ z^h_\al$ of
$P^{h'}\to X$. Then it is readily observed that
\mar{510ff2}\beq
\vr^{h'}_{\al\bt}(x) =\vr^h_{\al\bt}(x), \qquad x\in U_\al\cap
U_\bt. \label{510ff2}
\eeq
Conversely, any isomorphism $(\Phi, \id X)$ of reduced principal
subbundles $P^h$ and $P^{h'}$ of $P$ defines an $H$-equivariant
$G$-valued function $f$ on $P^h$ given by the relation
\be
pf(p)=\Phi(p), \qquad p\in P^h.
\ee
Its prolongation to a $G$-equivariant function on $P$ is defined
as
\be
f(pg)=g^{-1}f(p)g, \qquad p\in P^h, \qquad g\in G.
\ee
In accordance with the relation (\ref{b3111}), this function
provides a vertical principal automorphism of $P$ whose
restriction to $P^h$ coincides with $\Phi$.
\end{proof}

\begin{theorem}\label{iso0} \mar{iso0} If the quotient $G/H$ is homeomorphic to a
Euclidean space $\mathbb R^m$, all principal $H$-subbundles of a
principal $G$-bundle $P$ are isomorphic to each other \cite{ste}.
\end{theorem}

\begin{remark} \mar{510r50} \label{510r50} Let $P^h$ and $P^{h'}$
be isomorphic reduced principal subbundles in Theorem \ref{iso1}.
A principal $G$-bundle $P$ provided with the atlas $\Psi^h$
(\ref{510ff}) can be regarded as a $P^h$-associated bundle with a
structure group $H$ acting on its typical fibre $G$ on the left.
Endowed with the atlas $\Psi^{h'}$ (\ref{510ff1}), it is a
$P^{h'}$-associated $H$-bundle. The $H$-bundles $(P,\Psi^h)$ and
$(P,\Psi^{h'})$ fail to be equivalent because their atlases
$\Psi^h$ and $\Psi^{h'}$ are not equivalent. Indeed, the union of
these atlases is an atlas
\be
\Psi=\{(U_\al,z^h_\al, z^{h'}_\al), \vr^h_{\al\bt},
\vr^{h'}_{\al\bt}, \vr_{\al\al}=f(z_\al)\}
\ee
possessing transition functions
\mar{510f11}\beq
z^{h'}_\al =z^h_\al\vr_{\al\al}, \qquad
\vr_{\al\al}(x)=f(z_\al(x)), \label{510f11}
\eeq
between the bundle charts $(U_\al,z^h_\al)$ and
$(U_\al,z^{h'}_\al)$ of $\Psi^h$ and $\Psi^{h'}$, respectively.
However, the transition functions $\vr_{\al\al}$ are not
$H$-valued. At the same time, a glance at the equalities
(\ref{510ff2}) shows that transition functions of both the atlases
form the same cocycle. Consequently, the $H$-bundles $(P,\Psi^h)$
and $(P,\Psi^{h'})$ are associated. Due to the isomorphism
$\Phi:P^h\to P^{h'}$, one can write
\be
&& P=(P^h\times G)/H=(P^{h'}\times G)/H,\\
&& (p\times g)/H=(\Phi(p)\times f^{-1}(p)g)/H.
\ee
For any $\rho\in H$, we have
\be
&&(p\rho,g)/H=(\Phi(p)\rho, f^{-1}(p)g)/H=(\Phi(p),\rho
f^{-1}(p)g)/H = \\
&& \qquad  (\Phi(p), f^{-1}(p)\rho'g)/H,
\ee
where
\mar{510f23}\beq
\rho'= f(p)\rho f^{-1}(p). \label{510f23}
\eeq
It follows that $(P,\Psi^{h'})$ can be regarded as a
$P^h$-associated bundle with the same typical fibre $G$ as that of
$(P,\Psi^h)$, but the action $g\to\rho' g$ (\ref{510f23}) of a
structure group $H$ on a typical fibre of $(P,\Psi^{h'})$ is not
equivalent to its action $g\to\rho g$ on a typical fibre of
$(P,\Psi^h)$ since they possesses different orbits in $G$.
\end{remark}

Given a classical Higgs field $h$ and the corresponding reduced
principal $H$-bundle $P^h$, let
\mar{510f24}\beq
Y^h=(P^h\times V)/H \label{510f24}
\eeq
be the associated vector bundle with a typical fibre $V$ which
admits a representation of a group $H$ of exact symmetries. Its
sections $s_h$ describe matter fields in the presence of a
classical Higgs field $h$ (Remark \ref{ss40}).

In general, the fibre bundle $Y^h$ (\ref{510f24}) fails to be
associated to another principal $H$-subbundles $P^{h'}$ of $P$. It
follows that, in this case, a $V$-valued matter field can be
represented only by a pair with a certain Higgs field. Therefore,
a goal is to describe the totality of these pairs $(s_h,h)$ for
all Higgs fields $h\in \Si(X)$.

\begin{remark} \label{510r90} \mar{510r90}
If reduced principal $H$-subbundles $P^h$ and $P^{h'}$ of a
principal $G$-bundle are isomorphic in accordance with Theorem
\ref{iso1}, then the $P^h$-associated bundle $Y^h$ (\ref{510f24})
is associated as
\mar{510f25}\beq
Y^h=(\Phi(p)\times V)/H \label{510f25}
\eeq
to $P^{h'}$. If a typical fibre $V$ admits an action of the whole
group $G$, the $P^h$-associated bundle $Y^h$ (\ref{510f24}) also
is $P$-associated as
\be
Y^h=(P^h\times V)/H= (P\times V)/G.
\ee
Such $P$-associated bundles $P^h$ and $P^{h'}$ are equivalent as
$G$-bundles, but they fail to be equivalent as $H$-bundles because
transition functions between their atlases are not $H$-valued
(Remark \ref{510r50}).
\end{remark}

In order to describe matter fields in the presence of different
classical Higgs fields, let us consider the composite bundle
(\ref{b3223a}) and the composite bundle
\mar{b3225}\beq
Y\ar^{\pi_{Y\Si}} \Si\ar^{\pi_{\Si X}} X \label{b3225}
\eeq
where $Y\to \Si$ is a $P_\Si$-associated bundle
\mar{bnn}\beq
Y=(P\times V)/H \label{bnn}
\eeq
with a structure group $H$. Given a Higgs field $h$ and the
corresponding reduced principal $H$-subbundle $P^h=h^*P$, the
$P^h$-associated fibre bundle (\ref{510f24}) is the restriction
\mar{b3226}\beq
Y^h=h^*Y=(h^*P\times V)/H \label{b3226}
\eeq
of a fibre bundle $Y\to\Si$ to $h(X)\subset \Si$. Every global
section $s_h$ of the fibre bundle $Y^h$ (\ref{b3226}) is a global
section of the composite bundle (\ref{b3225}) projected onto a
Higgs field $h=\pi_{Y\Si}\circ s$. Conversely, every global
section $s$ of the composite bundle $Y\to X$ (\ref{b3225})
projected onto a Higgs field $h=\pi_{Y\Si}\circ s$ takes its
values into the subbundle $Y^hY$ (\ref{b3226}). Hence, there is
one-to-one correspondence between the sections of the fibre bundle
$Y^h$ (\ref{510f24}) and the sections of the composite bundle
(\ref{b3225}) which cover $h$.

Thus, it is the composite bundle $Y\to X$ (\ref{b3225}) whose
sections describe the above mentioned totality of pairs $(s_h, h)$
of matter fields and Higgs fields in classical gauge theory with
spontaneous symmetry breaking \cite{book09,higgs,tmp}.

A key point is that, though $Y\to\Si$ is a fibre bundle with a
structure group $H$, a composite bundle $Y\to X$ is a
$P$-associated bundle as follows \cite{book09,tmp}.

\begin{theorem} \mar{LL1} \label{LL1} The composite bundle $Y\to X$ (\ref{b3225})
is a $P$-associated bundle
\be
&& Y=(P\times (G\times V)/H)/G, \\
&& (pg',(g\rho,v))= (pg',(g,\rho v))=(p,g'(g,\rho v))=(p,(g'g,\rho v)).
\ee
with a structure group $G$. Its typical fibre is a fibre bundle
\mar{wes}\beq
\pi_{WH}:W=(G\times V)/H\to G/H \label{wes}
\eeq
associated to a principal $H$-bundle $G\to G/H$ (\ref{ggh}). A
structure group $G$ acts on $W$ by the law
\mar{iik}\beq
g': (G\times V)/H \to (g'G\times V)/H. \label{iik}
\eeq
\end{theorem}

\begin{theorem} \mar{LL2} \label{LL2}
Given a Higgs field $h$, any atlas of a $P_\Si$-associated bundle
$Y\to \Si$ defines an atlas of a $P$-associated bundle $Y\to X$
with $H$-valued transition functions. The converse need not be
true.
\end{theorem}

\begin{proof} Any atlas $\Psi_{Y\Si}$ of a $P_\Si$-associated bundle $Y\to \Si$
is defined by an atlas
\mar{aaq0}\beq
\Psi_{P\Si}=\{(U_{\Si \iota},z_\iota), \vr_{\iota\kappa}\}
\label{aaq0}
\eeq
of the principal $H$-bundle $P_\Si$ (\ref{b3194}). Given a section
$h$ of $\Si\to X$, we have an atlas
\mar{aaq}\beq
\Psi^h=\{(\pi_{P\Si}(U_{\Si \iota}),z_\iota\circ h),
\vr_{\iota\kappa}\circ h\} \label{aaq}
\eeq
of the reduced principal $H$-bundle $P^h$ which also is an atlas
of $P$ with $H$-valued transition functions (Remark \ref{rttf}).
\end{proof}

Given an atlas $\Psi_P$ of $P$, the quotient bundle $\Si\to X$
(\ref{b3193}) is endowed with the associated atlas (\ref{aaq1}).
With this atlas and an atlas $\Psi_{Y\Si}$ of $Y\to \Si$, the
composite bundle $Y\to X$ (\ref{b3225}) is provided with adapted
bundle coordinates $(x^\la,\si^m,y^i)$ where $(\si^m)$ are fibre
coordinates on $\Si\to X$ and $(y^i)$ are those on $Y\to\Si$.

\begin{theorem} \mar{LL3} \label{LL3}
Any principal automorphism of a principal $G$-bundle $P\to X$ is
$G$-equivariant and,  consequently, $H$-equivariant. Thus, it is a
principal automorphism of a principal $H$-bundle $P\to\Si$ and,
consequently, it yields an automorphism of the $P_\Si$-associated
bundle $Y$ (\ref{b3225}).
\end{theorem}

The converse is not true. For instance, a vertical principal
automorphism of $P\to\Si$ is never a principal automorphism of
$P\to X$.

Theorems \ref{LL1} -- \ref{LL3} enables  one to describe matter
fields with an exact symmetry group $H\subset G$ in the framework
of gauge theory on a $G$-principal bundle $P\to X$ if its
structure group $G$ is reducible to $H$.

\subsection{Lorentz reduced structures in gravitation theory}

As was mentioned in Section 1, gravitation theory based on
Relativity and Equivalence Principles is formulated as gauge
theory on natural bundles (Remark \ref{ss45}) over a world
manifold whose structure group $GL_4$ (\ref{gl4}) is reduced to a
Lorentz subgroup $SO(1,3)$ \cite{book09,sard02,sard11}.

Natural bundles are exemplified by tensor bundles $T$ and, in
particular, the tangent bundle $TX$ over $X$. Given a
diffeomorphism $f$ of $X$, the tangent morphism $Tf:TX\to TX$  is
a general covariant transformation of $TX$. Tensor bundles over an
oriented world manifold possess the structure group $GL_4$
(\ref{gl4}). An associated principal bundle is the above mentioned
frame bundle $LX$ (Remark \ref{ss45}). Its (local) sections are
called frame fields. Given the holonomic atlas of the tangent
bundle $TX$, every element $\{H_a\}$ of a frame bundle $LX$ takes
a form $H_a=H^\m_a\dr_\m$, where $H^\m_a$ is a matrix of the
natural representation of a group $GL_4$ in $\mathbb R^4$. These
matrices constitute bundle coordinates
\be
(x^\la, H^\m_a), \qquad H'^\m_a=\frac{\dr x'^\m}{\dr
x^\la}H^\la_a,
\ee
on $LX$ associated to its holonomic atlas
\mar{tty}\beq
\Psi_\mathrm{T}=\{(U_\iota, z_\iota=\{\dr_\m\})\}, \label{tty}
\eeq
given by local frame fields $z_\iota=\{\dr_\m\}$. With respect to
these coordinates, the canonical right action of $GL_4$ on $LX$
reads $GL_4\ni g: H^\m_a\to H^\m_bg^b{}_a$.

A frame bundle $LX$ is equipped with a canonical $\mathbb
R^4$-valued one-form
\mar{b3133'}\beq
\thh_{LX} = H^a_\m dx^\m\ot t_a,\label{b3133'}
\eeq
where $\{t_a\}$ is a fixed basis for $\mathbb R^4$ and $H^a_\m$ is
the inverse matrix of $H^\m_a$.

A frame bundle $LX\to X$ is natural. Any diffeomorphism $f$ of $X$
gives rise to a principal automorphism
\mar{025}\beq
\wt f: (x^\la, H^\la_a)\to (f^\la(x),\dr_\m f^\la H^\m_a)
\label{025}
\eeq
of $LX$ which is its general covariant transformation.

Let $Y=(LX\times V)/GL_4$ be an $LX$-associated bundle with a
typical fibre $V$. It admits a lift of any diffeomorphism $f$ of
its base to an automorphism
\be
f_Y(Y)=(\wt f(LX)\times V)/GL_4
\ee
of $Y$ associated with the principal automorphism $\wt f$
(\ref{025}) of a frame bundle $LX$. Thus, all bundles associated
to a frame bundle $LX$ are natural bundles.

As was mentioned above, gravitation theory on a world manifold $X$
is a gauge theory with spontaneous symmetry breaking described by
Lorentz reduced structures of a frame bundle $LX$. We deal with
the following Lorentz and proper Lorentz reduced structures.

By a Lorentz reduced structure is meant a reduced
$SO(1,3)$-subbundle $L^gX$, called the Lorentz subbundle, of a
frame bundle $LX$.

Let L$=SO^0(1,3)$ be a proper Lorentz group. Recall that
$SO(1,3)=\mathbb Z_2\times$L. A proper Lorentz reduced structure
is defined as a reduced L-subbundle $L^hX$ of $LX$.

\begin{theorem} \label{ss48} \mar{ss48}
If a world manifold $X$ is simply connected, there is one-to-ne
correspondence between the Lorentz and proper Lorentz reduced
structures.
\end{theorem}

One can show that different proper Lorentz subbundles $L^hX$ and
$L^{h'}X$ of a frame bundle $LX$ are isomorphic as principal
L-bundles. This means that there exists a vertical automorphism of
a frame bundle $LX$ which sends $L^hX$ onto $L^{h'}X$
\cite{book09,higgs}. If a world manifold $X$ is simply connected,
the similar property of Lorentz subbundles also is true in
accordance with Theorem \ref{ss48}.

\begin{remark} \label{ss49} \mar{ss49}
There is the well-known topological obstruction to the existence
of a Lorentz structure on a world manifold $X$. All non-compact
manifolds and compact manifolds whose Euler characteristic equals
zero admit a Lorentz reduced structure \cite{dods}.
\end{remark}

By virtue of Theorem \ref{redsub}, there is one-to-one
correspondence between the principal L-subbundles $L^hX$ of a
frame bundle $LX$ and the global sections $h$ of a quotient bundle
\mar{5.15}\beq
\Si_\mathrm{T}=LX/\rL,  \label{5.15}
\eeq
called the tetrad bundle. This is an $LX$-associated bundle with
the typical fibre $GL_4/$L. Its global sections are called the
tetrad fields. The fibre bundle (\ref{5.15}) is a two-fold
covering $\zeta: \Si_{\rm T}\to \Si_{\rm PR}$ of the quotient
bundle
\mar{b3203}\beq
\Si_\mathrm{PR}=LX/SO(1,3) \label{b3203}.
\eeq
whose global sections $g$ are pseudo-Riemannian metrics of
signature $(+,---)$ on a world manifold/ It is called the metric
bundle.

In particular, every tetrad field $h$ defines a unique
pseudo-Riemannian metric $g=\zeta\circ h$. For the sake of
convenience, one usually identifies a metric bundle with an open
subbundle of the tensor bundle $\Si_{\rm PR}\subset \op\vee^2 TX$.
Therefore, the metric bundle $\Si_{\rm PR}$ (\ref{b3203}) can be
equipped with bundle coordinates $(x^\la, \si^{\m\nu})$.

Every tetrad field $h$ defines an associated Lorentz bundle atlas
\mar{lat}\beq
\Psi^h=\{(U_\iota,z_\iota^h=\{h_a\})\} \label{lat}
\eeq
of a frame bundle $LX$ such that the corresponding local sections
$z_\iota^h$ of $LX$ take their values into a proper Lorentz
subbundle $L^hX$ and the transition functions of $\Psi^h$
(\ref{lat}) between the frames $\{h_a\}$ are L-valued. The frames
(\ref{lat}):
\mar{b3211a}\beq
\{h_a =h_a^\m(x)\dr_\m\}, \qquad h_a^\m=H_a^\m\circ z_\iota^h,
\qquad x\in U_\iota, \label{b3211a}
\eeq
are called the tetrad frames.

Given a Lorentz bundle atlas $\Psi^h$, the pull-back
\mar{b3211}\beq
h=h^a\ot t_a=z_\iota^{h*}\thh_{LX}=h_\la^a(x) dx^\la\ot t_a
\label{b3211}
\eeq
of the canonical form $\thh_{LX}$ (\ref{b3133'}) by a local
section $z_\iota^h$ is called the (local) tetrad form. It
determines tetrad coframes
\mar{b3211'}\beq
\{h^a =h^a_\m(x)dx^\m\}, \qquad x\in U_\iota, \label{b3211'}
\eeq
in the cotangent bundle $T^*X$. They are the dual of the tetrad
frames (\ref{b3211a}). The coefficients $h_a^\m$ and $h^a_\m$ of
the tetrad frames (\ref{b3211a}) and coframes (\ref{b3211'}) are
called the tetrad functions. They are transition functions between
the holonomic atlas $\Psi_\mathrm{T}$ (\ref{tty}) and the Lorentz
atlas $\Psi^h$ (\ref{lat}) of a frame bundle $LX$.

With respect to the Lorentz atlas $\Psi^h$ (\ref{lat}), a tetrad
field $h$ can be represented by the $\mathbb R^4$-valued tetrad
form (\ref{b3211}). Relative to this atlas, the corresponding
pseudo-Riemannian metric $g=\zeta\circ h$ takes the well-known
form
\mar{mos175}\beq
g=\eta(h\ot h)=\eta_{ab}h^a\ot h^b, \qquad
g_{\m\nu}=h_\m^ah_\nu^b\eta_{ab}, \label{mos175}
\eeq
where $\eta={\rm diag}(1,-1,-1,-1)$ is the Minkowski metric in
$\mathbb R^4$ written with respect to its fixed basis $\{t_a\}$.
It is readily observed that the tetrad coframes $\{h^a\}$
(\ref{b3211'}) and the tetrad frames $\{h_a\}$ (\ref{b3211a}) are
orthornormal relative to the pseudo-Riemannian metric
(\ref{mos175}), namely:
\be
g^{\m\nu}h^a_\mu h^b_\nu=\eta^{ab}, \qquad g_{\m\nu}h_a^\mu
h_b^\nu=\eta_{ab}.
\ee

\section{Spinor structures}

As was mentioned above, we aim to describe spinor bundles as
subbundles of a fibre bundle in complex Clifford algebras.

\subsection{Fibre bundles in Clifford algebras}

One usually consider fibre bundles in Clifford algebras whose
structure group is a group of automorphisms of these algebras
\cite{book09,law} (Remark \ref{sp103}). A problem is that this
group fails to preserve spinor subspaces of a Clifford algebra
and, thus, it can not be a structure group of spinor bundles.
Therefore, we define fibre bundles in Clifford algebras whose
structure group is a group of invertable elements of a Clifford
algebra which acts on this algebra by left multiplications.
Certainly, it preserves minimal left ideals of this algebra and,
consequently, it is a structure group of spinor bundles. In a case
in question, this is a matrix group.

Let $\mathbb C\cC(n)$ be a complex Clifford algebra modelled over
an even dimensional complex space $\mathbb C^n$ (Definition
\ref{ss19}). It is isomorphic to a ring $\mathrm{Mat}(2^{n/2},
\mathbb C)$ of complex $(2^{n/2}\times 2^{n/2})$-matrices (Theorem
\ref{k8}). Its invertible elements constitute a general linear
group $GL(2^{n/2}, \mathbb C)$ whose adjoint representation in
$\mathbb C\cC(n)$ yields the projective linear group $PGL(2^{n/2},
\mathbb C)$ (\ref{k25}) of automorphisms of $\mathbb C\cC(n)$
(Theorem \ref{k21}).

\begin{definition} \label{ps101} \mar{ps101}
Given a smooth manifold $X$, let us consider a principal bundle
$P\to X$ with a structure group $GL(2^{n/2}, \mathbb C)$. A fibre
bundle in complex Clifford algebras $\mathbb C\cC(n)$ is defined
to be the $P$-associated bundle (\ref{b1.230}):
\mar{sp100}\beq
\lC= (P\times \mathrm{Mat}(2^{n/2}, \mathbb C))/GL(2^{n/2},
\mathbb C)\to X \label{sp100}
\eeq
with a typical fibre
\mar{sp102}\beq
\mathbb C\cC(n)= \mathrm{Mat}(2^{n/2}, \mathbb C) \label{sp102}
\eeq
which carries out the left-regular representation of a group
$GL(2^{n/2}, \mathbb C)$.
\end{definition}

Owing to the canonical inclusion $GL(2^{n/2}, \mathbb C)\to
\mathrm{Mat}(2^{n/2}, \mathbb C)$, a principal $GL(2^{n/2},
\mathbb C)$-bundle $P$ is a subbundle $P\subset \lC$ of the
Clifford algebra bundle $\lC$ (\ref{sp100}). Herewith, the
canonical right action of a structure group $GL(2^{n/2}, \mathbb
C)$ on a principal bundle $P$ is extended to the fibrewise action
of $GL(2^{n/2}, \mathbb C)$ on the Clifford algebra bundle $\lC$
(\ref{sp100}) by right multiplications. This action is globally
defined because it is commutative with transition functions of
$\lC$ acting on its typical fibre $\mathrm{Mat}(2^{n/2}, \mathbb
C)$ on the left.

\begin{remark} \label{sp103} \mar{sp103}
As was mentioned above, one usually considers a fibre bundle in
Clifford algebras $\mathrm{Mat}(2^{n/2})$ whose structure group is
the projective linear group $PGL(2^{n/2}, \mathbb C)$ (\ref{k25})
of automorphisms of $\mathbb C\cC(n)$. This also is a
$P$-associated bundle
\mar{ps104}\beq
\cA\lC=(P\times \mathbb C\cC(n))/GL(2^{n/2}, \mathbb C)\to X
\label{sp104}
\eeq
where $GL(2^{n/2}, \mathbb C)$ acts on $\mathbb C\cC(n)$ by the
adjoint representation. In particular, a certain subbundle of
$\cA\lC$ (\ref{sp104}) is the group bundle $P^G$ (\ref{b3130})
(Remark \ref{56e2}).
\end{remark}

Let $\Psi(n)$ (Definition \ref{ss34}) be a spinor space of a
complex Clifford algebra $\mathbb C\cC(n)$. Being a minimal left
ideal of $\mathbb C\cC(n)$ (Definition \ref{ss33}), it is a
subspace $\Psi(n)$ of $\mathbb C\cC(n)$ (Theorem \ref{k53}) which
inherits the left-regular representation of a group $GL(2^{n/2},
\mathbb C)$ in $\mathbb C\cC(n)$.

\begin{definition} \label{sp106} \mar{sp106}
Given a principal $GL(2^{n/2}, \mathbb C)$-bundle $P$, a spinor
bundle is defined as a $P$-associated bundle
\mar{sp107}\beq
S=(P\times \Psi(n))/GL(2^{n/2}, \mathbb C)\to X  \label{sp107}
\eeq
with a typical fibre $\Psi(n)=\mathbb C^{2^{n/2}}$ and a structure
group $GL(2^{n/2}, \mathbb C)$ which acts on $\Psi(n)$ by left
multiplications  that is equivalent to the natural matrix
representation of $GL(2^{n/2},\mathbb C)$ in $\mathbb C^{2^{n/2}}$
(Corollary \ref{k54}).
\end{definition}

Obviously, the spinor bundle $S$ (\ref{sp107}) is a subbundle of
the Clifford algebra bundle $\lC$ (\ref{sp100}). However, $S$
(\ref{sp107}) need not  be a subbundle of the fibre bundle
$\cA\lC$ (\ref{sp104}) in Clifford algebras because a spinor space
$\Psi(n)$ is not stable under automorphisms of a Clifford algebra
$\mathbb C\cC(n)$.

At the same time, given the spinor representation  (\ref{sp110}),
of a complex Clifford algebra, there is a fibrewise morphism
\mar{sp111}\ben
&& \g: \cA\lC\op\times_X S\ar_X S,  \label{sp111}\\
&& (P\times (\mathbb C\cC(n)\times \Psi(n)))/GL(2^{n/2}, \mathbb C)
\to \nonumber \\
&& \qquad  (P\times \g(\mathbb C\cC(n)\times \Psi(n)))/GL(2^{n/2},
\mathbb C), \nonumber
\een
of the $P$-associated fibre bundles $\cA\lC$ (\ref{sp104}) and $S$
(\ref{sp107}) with a structure group $GL(2^{n/2}, \mathbb C)$.

\begin{remark} \label{sp210} \mar{sp210} Let $X$ be a smooth real
manifold of dimension $2^{n/2}$, $n=2,4,\ldots$. Let $TX$ be the
tangent bundle over $X$ and $LX$ the associated principal frame
bundle. Their structure group is $GL(2^{n/2}, \mathbb R)$. There
is the canonical group monomorphism
\mar{sp211}\beq
GL(2^{n/2}, \mathbb R) \to GL(2^{n/2}, \mathbb C). \label{sp211}
\eeq
Let us consider a trivial complex line bundle $X\times\mathbb C$
over $X$ and a complexification
\mar{sp212}\beq
{\mathbb C}TX=(X\times \mathbb C)\op\ot_X TX=\mathbb C\op\ot_X TX
\label{sp212}
\eeq
of $TX$. This is a fibre bundle with a structure group
$GL(2^{n/2}, \mathbb C)$ which is reducible to its subgroup
$GL(2^{n/2}, \mathbb R)$ (\ref{sp211}). Let $P\to X$ be an
associated principal $GL(2^{n/2}, \mathbb C)$-bundle. There is the
corresponding mnomorphism of principal bundles
\mar{sp213}\beq
LX\ar_X P. \label{sp213}
\eeq
Let $\lC$ be the $P$-associated bundle (\ref{sp100}) in complex
Clifford algebras $\mathbb C\cC(n)$. Let $S$ (\ref{sp107}) be a
spinor subbundle of $\lC$ whose typical fibre $\mathbb
C^{2^{n/2}}$ carries out the natural matrix representation of
$GL(2^{n/2},\mathbb C)$. Due to the monomorphism (\ref{sp213}),
its structure group is reducible to $GL(2^{n/2},\mathbb R)$, and
$S$ is a $LX$-associated bundle isomorphic to the complex tangent
bundle ${\mathbb C}TX$ (\ref{sp212}). Thus, a complexification of
the tangent bundle of a smooth real manifold of dimension
$2^{n/2}$, $n=2,4,\ldots$ can be represented as a spinor bundle.
Moreover, general covariant transformations of $LX$ gives rise to
principal automorphisms of a principal bundle $P$ (Remark
\ref{sp220}). Therefore, a principal bundle $P$, a $P$-associated
Clifford algebra bundle $\lC$ and its spinor subbundle $S$ are
natural bundles (Remark \ref{ss45}).
\end{remark}

It should be emphasized that, though there is the ring
monomorphism $\cC(m,n-m)\to \mathbb C\cC(n)$ (\ref{sp200}), the
Clifford algebra bundle $\lC$ (\ref{sp100}) need not contains a
subbundle in real Clifford algebras $\cC(m,n-m)$, unless a
structure group $GL(2^{n/2}, \mathbb C)$ of $\lC$ is reducible to
a group $\cG\cC(m,n-m)$. This problem can be solved as follows.

Let $X$ be a smooth real manifold of even dimension $n$. Let $TX$
be the tangent bundle over $X$ and $LX$ the associated principal
frame bundle. Let us assume that their structure group is $GL(n,
\mathbb R)$ is reducible to a pseudo-ortohogonal subgroup
$O(m,n-m)$. In accordance with Theorem \ref{comp}, a structure
group $GL(n, \mathbb R)$ always is reducible to a subgroup
$O(n,\mathbb R)$. There is the exact sequence of groups
(\ref{106}):
\mar{sp121}\beq
 e\to \mathbb Z_2\longrightarrow
\mathrm{Pin}(m,n-m)\op\longrightarrow^\zeta O(m,n-m)\to e.
\label{sp121}
\eeq
A problem is that this exact sequence is not split, i.e., there is
no monomorphism $\kappa: O(m,n-m)\to \mathrm{Pin}(m,n-m)$ so that
$\zeta\circ\kappa=\id$ (Example \ref{ss6} and Remark \ref{sp225}).

In this case, we say that a principal $\mathrm{Pin}(m,n-m)$-bundle
$\wt P^h\to X$ is an extension of a principal $O(m,n-m)$-bundle
$P^h\to X$ if there is an epimorphism of principal bundles
\mar{sp223}\beq
\wt P^h\to P^h \label{sp223}
\eeq
(Remark \ref{52r1'}). Such an extension need not exist. The
following is a corollary of the well-known theorem
\cite{book09,greub,law}.

\begin{theorem} \label{sp222} \mar{sp222}
The topological obstruction to that a principal $O(m,n-m)$-bundle
$P^h\to X$ lifts to a principal $\mathrm{Pin}(m,n-m)$-bundle $\wt
P^h\to X$ is given by the \v Cech cohomology group $H^2(X;\mathbb
Z_2)$ of $X$. Namely, a principal bundle $P$ defines an element of
$H^2(X;\mathbb Z_2)$ which must be zero so that $P^h\to X$ can
give rise to $\wt P^h\to X$. Inequivalent lifts of $P^h\to X$ to
principal $\mathrm{Pin}(m,n-m)$-bundles are classified by elements
of the \v Cech cohomology group $H^1(X;\mathbb Z_2)$.
\end{theorem}

Let $L^hX$ be a reduced principal $O(m,n-m)$-subbundle of a frame
bundle. In this case, the topological obstruction in Theorem
\ref{sp222} to that this bundle $L^hX$ is extended to a principal
$\mathrm{Pin}(m,n-m)$-bundle $\wt L^hX$ is the second
Stiefel--Whitney class $w_2(X)\in H^2(X;\mathbb Z_2)$ of $X$
\cite{law}. Let us assume that a manifold $X$ is orientable, i.e.,
the \v Cech cohomology group $H^1(X;\mathbb Z_2)$ is trivial,  and
that the second Stiefel--Whitney class $w_2(X)\in H^2(X;\mathbb
Z_2)$ of $X$ also is trivial. Let (\ref{sp223}) be a desired
$\mathrm{Pin}(m,n-m)$-lift of a principal $O(m,n-m)$-bundle $P^h$.

Owing to the canonical monomorphism (\ref{sp200}) of Clifford
algebras, there is the group monomorphism $\mathrm{Pin}(m,n-m)\to
\cG\mathbb C\cC(n)$. Due to this monomorphism, there exists a
principal $\cG\mathbb C\cC(n)$-bundle $P$ whose reduced
$\mathrm{Pin}(m,n-m)$-subbundle is $\wt L^hX$, and whose structure
group $\cG\mathbb C\cC(n)$ thus is reducible to
$\mathrm{Pin}(m,n-m)$. Let
\mar{sp240}\beq
\lC^h\to X\label{sp240}
\eeq
be the $P$-associated Clifford algebra bundle (\ref{sp100}). Then
it contains a subbundle
\mar{sp241}\beq
\lC^h(m,n-m)\to X\label{sp241}
\eeq
in real Clifford algebras $\cC(m,n-m)$. This subbundle in turn
contains a subbundle of generating vector spaces which is
$L^hX$-associated and, thus, isomorphic to the tangent bundle $TX$
as a $L^hX$-associated bundle. The Clifford algebra bundle $\lC^h$
(\ref{sp240}) contains spinor subbundles $S^h\to X$ (\ref{sp107})
together with the representation morphisms (\ref{sp111}).

Of course, with a different reduced principal $O(m,n-m)$-subbundle
$L^{h'}X$ of $LX$, we come to a different Clifford algebra bundle
$\lC^{h'}$ (\ref{sp240}). Let us recall that, in accordance with
Theorem \ref{redsub}, there is one-to-one correspondence
(\ref{510f2}) between the reduced principal $O(m,n-m)$-subbundle
$L^hX$ of $LX$ and the global sections of the quotient bundle
\mar{sp226}\beq
\Si(m,n-m)=LX/O(m,n-m)\to X \label{sp226}
\eeq
which are pseudo-Riemannian metrics on $X$ of signature $(m,n-m)$.

\begin{remark} \label{sp225} \mar{sp225}
Let $X$ be a four-dimensional manifold. In this case, we have the
fibre bundle $\lC$ in Clifford algebras $\mathbb C\cC(4)$ in
Example \ref{sp210} and the Clifford algebra bundles $\lC^h$
(\ref{sp240}) in $\mathbb C\cC(4)$. A difference between them lies
in the fact that a structure group $GL(4,\mathbb C)$ of $\lC$ is
reducible to a subgroup $O(m,4-m)$, whereas the a structure group
of $\lC^h$ is reducible to $\mathrm{Pin}(m,4-m)$, but $O(m,4-m)$
is not a subgroup of $\mathrm{Pin}(m,4-m)$.
\end{remark}

A key point is that, given different sections $h$ and $h'$ of the
quotient bundle $\Si(m,n-m)$ (\ref{sp226}), the Clifford algebra
bundles $\lC^h$ and $\lC^{h'}$ need not be isomorphic as follows.
These fibre bundles are associated to principal
$\mathrm{Pin}(m,n-m)$-bundles $\wt L^hX$ and $\wt L^{h'}X$ which
are the two-fold covers (\ref{sp223}) of the reduced principal
$O(m,n-m)$-subbundles $L^hX$ and $L^{h'}X$ of a frame bundle $LX$,
respectively. These subbundles need not be isomorphic, and then
the principal bundles $\wt L^hX$, $\wt L^{h'}X$ and, consequently,
associated Clifford algebra bundles $\lC^h$, $\lC^{h'}$
(\ref{sp240}) are well. Moreover, let principal bundles $L^hX$ and
$L^{h'}X$ be isomorphic. For instance, in accordance with Theorem
\ref{iso0}, this is the case of an orthogonal group
$O(n,o)=O(n,\mathbb R)$. However, their covers $\wt L^hX$ and $\wt
L^{h'}X$ need not be isomorphic. Thus a Clifford algebra bundle
must be considered only in a pair with a certain pseudo-Riemannian
metric $h$.

\subsection{Composite bundles in Clifford algebras}

In order to describe a whole family of non-isomorphic Clifford
algebra bundles $\lC^h$, one can follow a construction of the
composite bundle (\ref{b3225}). Let us consider the composite
bundle (\ref{b3223a}):
\mar{sp230}\beq
LX \to \Si(m,n-m) \to X \label{sp230}
\eeq
where
\mar{sp231}\beq
LX\ar^{\pi_{P\Si}} \Si(m,n-m) \label{sp231}
\eeq
is a principal bundle with a structure group $O(m,n-m)$. Let us
consider its principal $\mathrm{Pin}(m,n-m)$-cover (\ref{sp223}):
\mar{sp232}\beq
\wt P_\Si\to \Si(m,n-m) \label{sp232}
\eeq
if it exists. Then, given a global section $h$ of $\Si(m,n-m) \to
X$ (\ref{sp226}), the pull-back $h^*\wt P_\Si$ is a subbundle of
$\wt P_\Si\to X$ which a $\mathrm{Pin}(m,n-m)$-cover
\be
h^*\wt P_\Si\to L^hX.
\ee

Owing to the canonical monomorphism (\ref{sp200}) of Clifford
algebras, there is the group monomorphism $\mathrm{Pin}(m,n-m)\to
\cG\mathbb C\cC(n)$. Due to this monomorphism, there exists a
principal $\cG\mathbb C\cC(n)$-bundle
\mar{sp245}\beq
\wt LX_\Si\to \Si(m,n-m). \label{sp245}
\eeq
Let
\mar{sp246}\beq
\lC_\Si\to \Si(m,n-m) \label{sp246}
\eeq
be the associated Clifford algebra bundle. It contains a subbundle
\mar{sp247}\beq
\lC_\Si(m,n-m)\to \Si(m,n-m)\label{sp247}
\eeq
and the spinor subbundles
\mar{sp248}\beq
S_\Si\to \Si(m,n-m). \label{sp248}
\eeq

\begin{theorem} \label{sp511} \mar{sp511}
Given a global section $h$ of $\Si(m,n-m) \to X$ (\ref{sp226}),
the pull-back bundles $h^*\lC_\Si\to X$, $h^*\lC_\Si(m,n-m)\to X$
and $h^*S_\Si\to X$ are subbundles of the composite bundles
$\lC_\Si\to X$, $\lC_\Si(m,n-m)\to X$ and $S_\Si\to X$ and are the
bundles $\lC^h\to X$ (\ref{sp240}), $\lC^h(m,n-m)\to X$
(\ref{sp241}) and $S^h\to X$, respectively.
\end{theorem}

As was mentioned above, this is just the case of gravitation
theory that we study in forthcoming Part II of our work.

\end{document}